\documentclass[12pt]{article}
\usepackage[margin=1.95cm]{geometry}

\usepackage{amsmath}
\usepackage{amsthm}
\usepackage{multirow}
\usepackage{times}
\usepackage{bm}
\usepackage{amssymb}
\usepackage{lscape}
\usepackage{color} 
\usepackage{soul} 
\usepackage{url}
\usepackage[figuresright]{rotating}
\usepackage{natbib}
\usepackage{setspace}
\usepackage{textcomp}
\usepackage{enumitem}

\newcommand{\mR}{\mathbb{R}}
\newcommand{\Pn}{{P}_n}
\newcommand{\bX}{\bm{X}}
\newcommand{\bI}{\bm{I}}
\newcommand{\bA}{\bm{A}}
\newcommand{\bY}{\bm{Y}}
\newcommand{\bxi}{\bm{\xi}}
\newcommand{\bbeta}{\bm{\beta}}
\newcommand{\bzeta}{\bm{\zeta}}
\newcommand{\bx}{\bm{x}}
\newcommand{\bW}{\bm{W}}
\newcommand{\bM}{\bm{M}}
\newcommand{\bTheta}{\bm{\Theta}}
\newcommand{\bSigma}{\bm{\Sigma}}

\newcommand{\E}{{E}}

\newcommand{\lambdamax}{\lambda_{\mathrm{max}}}
\newcommand{\lambdamin}{\lambda_{\mathrm{min}}}
\newcommand{\ba}{\bm{a}}
\newcommand{\bb}{\bm{b}}

\newcommand{\bZ}{\bm{Z}}
\newcommand{\bF}{\bm{F}}
\newcommand{\OP}{\mathcal{O}_{P}}
\newcommand{\oP}{o_{P}}
\newcommand{\balpha}{\bm{\alpha}}
\newcommand{\bDelta}{\bm{\Delta}}

\newcommand{\hTheta}{\widehat{\Theta}}

\newcommand{\Thetabeta}{\bm{\Theta}_{\bxi^0}}

\newtheorem{assumption}{Assumption}
\newtheorem{theorem}{Theorem}

\newtheorem{lemma}[theorem]{Lemma}
\newtheorem{remark}{Remark}

\begin{document}

\title{\bf De-biased Lasso for Generalized Linear Models with A Diverging Number of Covariates}

\author{Lu Xia\textsuperscript{1}, Bin Nan\textsuperscript{2*}, and Yi Li\textsuperscript{3*}\\
	\small 
	\textsuperscript{1}Department of Biostatistics, University of Washington, Seattle, WA, {xialu@uw.edu}  \\
	\small
	\textsuperscript{2}Department of Statistics, University of Californina, Irvine, CA, {nanb@uci.edu} \\
	\small
	\textsuperscript{3}Department of Biostatistics, University of Michigan, Ann Arbor, MI, {yili@umich.edu} \\
	\small *To whom correspondence should be addressed
}

\date{}

\maketitle

\begin{abstract}
Modeling and drawing inference on the joint associations between single nucleotide polymorphisms and a disease has sparked interest in genome-wide associations studies. In the motivating Boston Lung Cancer Survival Cohort (BLCSC) data, the presence of a large number of single nucleotide polymorphisms of interest, though smaller than the sample size, challenges inference on their joint associations with the disease outcome.  In similar settings, we find that neither the de-biased lasso approach  \citep{van2014asymptotically}, which assumes sparsity on the inverse information matrix, nor the standard maximum likelihood method can yield confidence intervals with satisfactory coverage probabilities for generalized linear models. Under this ``large $n$, diverging $p$" scenario, we propose an alternative de-biased lasso approach by directly inverting the Hessian matrix without imposing the matrix sparsity assumption, which further reduces bias compared to the original de-biased lasso and ensures valid confidence intervals with nominal coverage probabilities. We establish the asymptotic distributions of any linear combinations of the parameter estimates, which lays the theoretical ground for drawing inference. Simulations show that the proposed {\em refined} de-biased estimating method performs well in removing bias and yields honest confidence interval coverage. We use the proposed method to analyze the aforementioned  BLCSC data, a large scale hospital-based epidemiology cohort study, that investigates the joint effects of genetic variants on  lung cancer risks. \\[0.1cm]
	
\noindent \textbf{Keywords:} asymptotics, bias correction, high-dimensional regression,  lung cancer, statistical inference   
\end{abstract}



\section{Introduction \label{sec:intro}}

To identify disease-related genetic markers,
traditional genome-wide association studies typically analyze the marginal associations of the disease outcome with single nucleotide polymorphisms (SNPs), one at a time.  As  marginal associations do not  account for the dependence among SNPs,  false positive discoveries may occur as  SNPs can be claimed as significant when they are correlated with  the causal variants \citep{schaid2018genome}.  Alternatively,
modeling the joint effects of SNPs within the target genes
can reduce false positives around true causal SNPs and improve prediction accuracy \citep{he2010variable},  and  also 
can pinpoint  functionally impactful loci in the coding regions \citep{taylor2001using, repapi2010genome} so as to  better understand the molecular mechanisms underlying cancer \citep{guan2011bayesian}. { For example, among a subset of 1,374 patients from the Boston Lung Cancer Survival Cohort (BLCSC),  an epidemiology study that investigates molecular mechanisms underlying  lung cancer, our goal is to study the joint associations of  lung cancer risk  with over 100 SNPs residing in nine target genes that have been reported to harbor relevant genetic variants 
	\citep{mckay2017large}. The results may aid in personalized medicine by properly implicating relevant genetic variants and their joint roles in pharmacogenomics \citep{evans2004moving}. Statistically, the analysis requires reliable estimation and inference on a fairly large number of regression parameters.}

{With lung cancer mechanisms  differing by smoking  predisposition  \citep{bosse2018decade}, analyzing  BLCSC  among the 1,077  smokers and 297  non-smokers, separately, is necessary.
	Included in our  models are  103 SNPs and 4 demographic variables, which, though smaller than the number of  smokers or non-smokers, are  large enough to defy the conventional maximum likelihood estimation (MLE) approach. In particular, for non-smokers, Table \ref{tab:nonsmoker} in Section \ref{sec:app} has shown unreasonably large MLE estimates with wide confidence intervals, e.g. a point estimate of -19.64 with a 95\% confidence interval (-6,705.04, 6,665.75) for SNP AX-62479186.  Failures of  MLE  in similar scenarios have  been documented in \citet{sur2019modern},} and further evidenced by our {later} simulation studies.

The asymptotic framework  underlying these cases can be characterized as the number of parameters $p$ increasing with the sample size $n$, rather than staying fixed, which is often referred to as the ``large $n$, diverging $p$" scenario.
Drawing inference with generalized linear models (GLMs) under this  framework may facilitate a  range of applications, because the setting enables us to build valid models when the collected information increases with more subjects included in the study  \citep{wang2011gee}. {Several authors \citep{huber1973robust,yohai1979asymptotic,portnoy1984asymptotic,portnoy1985asymptotic} investigated the relative order between $p$ and $n$ that ensures the validity of M-estimators in linear regression;} 
\citet{he2000parameters} studied the consistency and the asymptotic normality of the M-estimators under different conditions and showed that $p^2 \log(p) / n \to 0$ would be needed for linear and logistic regression; \citet{wang2011gee} developed an asymptotic theory for the estimated regression parameters from generalized estimating equations with clustered binary outcomes, provided $p^3/n \rightarrow 0$. 
However, most of these methods incur substantial biases in empirical studies unless $p$ is very small.

Penalized regression methods have been developed over the decades to accommodate a large number of covariates. These methods, including the lasso  \citep{tibshirani1996regression}, the elastic net  \citep{zou2005regularization} and the Dantzig selector \citep{candes2007dantzig} among many others,  are  considered to be useful alternatives to the traditional variable selection methods such as forward or stepwise selection, especially in genetic studies \citep{schaid2018genome}.  These regularized methods yield  biased estimates, and, thus, cannot be directly used for drawing inference such as constructing confidence intervals with a nominal coverage probability.  

One stream of inferential methods is the post-selection inference conditional on selected models \citep{lee2016exact}, which requires conditional coverage to quantify the uncertainty associated with model selection. Other super-efficient procedures, such as SCAD \citep{fan2001variable, fan2004nonconcave} and adaptive lasso \citep{zou2006adaptive}, share the flavor of post-selection inference that is not the focus of this article. In particular, the inference based on the oracle estimation of \cite{fan2004nonconcave} requires  $p^5/n \rightarrow 0$.

Another school of methods is to draw inference by de-biasing the lasso estimates, termed de-biased lasso or de-sparsified lasso, which  relieves the restrictions of post-selection inference and  possesses nice theoretical and numerical properties in linear regression models (\citealt{van2014asymptotically, zhang2014confidence, javanmard2014confidence}). 


\citet{van2014asymptotically} extended  de-biased lasso  to GLMs and developed the asymptotic normality theory for each component of the coefficient estimates; based on this work, \citet{zhang2017simultaneous} proposed a multiplier bootstrap procedure
to draw inference on a group of coefficients in GLMs.
However, the de-biased lasso approach presented  subpar  performance with non-negligible biases and poor coverage of confidence intervals, {as seen from Figures \ref{fig:logit_n1k_ar1} and \ref{fig:logit_n1k_cs} for a logistic example in Section 4 that mimics the BLCSC setting,  because a key sparsity assumption on the inverse information matrix may not hold in GLM settings.}

To address the limitation and for valid inference with GLMs, we propose a {\em refined} de-biased lasso estimating method specifically tailored to the ``large $n$, diverging $p$" scenario as in {the motivating BLCSC dataset}. Our proposed method  estimates the inverse information matrix by directly inverting the sample Hessian matrix, which requires no structural assumptions on the inverse information matrix. We establish the asymptotic distributions for any linear combinations of the resulting estimates, laying the theoretical foundation for applications. Simulations demonstrate its better performance in reducing  biases and preserving confidence interval coverage probabilities than the conventional MLE and the original de-biased lasso \citep{van2014asymptotically} for a wide range of $p / n$ ratios, and all three methods yield almost identical results when $p$ is rather small relative to $n$. 

The rest of this article is organized as follows. Section~\ref{sec:method} describes in detail the model setup and the proposed {\em refined} de-biased lasso estimating method. Asymptotic results for the proposed method are provided in Section~\ref{sec:theory}, followed by simulation studies in Section~\ref{sec:sim}. Findings on the joint associations between SNPs in target genes and lung cancer risks by applying the proposed method to the motivating BLCSC data are reported in Section~\ref{sec:app}.
Not  to deviate from the main flow,  we put off the discussion of the distinctions of the proposed method from the existing high-dimensional inference literature to Section~\ref{sec:discuss}.

\section{Method \label{sec:method}}

\subsection{Background and set-up in generalized linear models \label{subsec:setup}}

We start with some  commonly used notation. For a vector $\ba$, $\|\ba\|_q$ denotes its $\ell_q$ norm, $q \ge 1$. Denote by $\lambdamax(\bA)$ and $\lambdamin(\bA)$  the largest and the smallest eigenvalues of a symmetric matrix $\bA$, respectively.
For a real matrix $\bA = (A_{ij})$, let $\| \bA \| = \sup_{\| \bx \|_2 = 1} \| \bA \bx \|_2 = [\lambdamax(\bA^T \bA)]^{1/2} $ be the spectral norm of $\bA$. The induced matrix $\ell_1$ norm is $\|\bA\|_1 = \max_j \sum_i  |A_{ij}|$, and  when $\bA$ is symmetric,  $\|\bA\|_1 = \max_i \sum_j  |A_{ij}|$ also holds. The element-wise $\ell_{\infty}$ norm is $\|\bA\|_{\infty} = \max_{i,j} |A_{ij}|$.  
With two positive sequences $a_n$ and $b_n$, write $a_n=\mathcal{O}(b_n)$ if there exist $c>0$ and $N>0$ such that $a_n<cb_n$ for all $n>N$, {and $a_n=o(b_n)$ if $a_n/b_n\rightarrow 0$ as $n\rightarrow\infty$}. We write $a_n\asymp b_n$ if 
$a_n=\mathcal{O}(b_n)$ and $b_n=\mathcal{O}(a_n)$.

Denote by $y_i$ the response variable and $\bx_i = (1, \widetilde{\bx}_i^T)^T \in \mathbb{R}^{p+1}$ for $i = 1,\ldots, n$, where 
``1"  corresponds to the intercept term, and $\widetilde{\bx}_i$ represents the $p$ covariates. Let $\bX$ be the $n \times (p+1)$ covariate matrix with $\bx_i^T$ being the $i$th row.
We assume  that $\left\{(y_i, \bx_i)\right\}_{i=1}^n$ are independent and identically distributed  copies of $(y,\bx)$. 
Define  the {\em negative} log-likelihood function as the following, up to an additive constant irrelevant to the unknown parameters, when the conditional density of $y$ given $\bx$  belongs to an exponential family:
\begin{equation} \label{eq:loss_func}
	\rho_{\bxi}(y, \bx) = \rho(y, \bx^T \bxi) = - y \bx^T \bxi + b(\bx^T \bxi)
\end{equation}
where $b(\cdot)$ is a known twice continuously differentiable function, $\bxi = (\beta_0, \bbeta^T)^T \in \mR^{p+1}$ denotes the vector of  coefficients and  $\beta_0 \in \mR$ is the intercept parameter. The unknown true coefficient vector is  $\bxi^0 = (\beta_0^0, {\bbeta^{0}}^T)^T$.

\subsection{De-biased lasso \label{subsec:debias}}

With  $\rho_{\bxi}(y, \bx) = \rho(y, \bx^T \bxi)$ given in (\ref{eq:loss_func}),  denote  by $\dot{{\rho}}_{\bxi}$ and $\ddot{{\rho}}_{\bxi}$ its first and second order derivatives with respect to $\bxi$, respectively. For any function $g(y, \bx)$, let $\Pn g = n^{-1} \sum_{i=1}^n g(y_i, \bx_i)$. Then for any $\bxi \in \mathbb{R}^{p+1}$, we denote the empirical loss function based on the random sample $\{ (y_i, \bx_i) \}_{i=1}^n$ by $\Pn\rho_{\bxi} = n^{-1} \sum_{i=1}^n \rho_{\bxi}(y_i, \bx_i)$, and its first and second order derivatives with respect to $\bxi$ by $\Pn \dot{{\rho}}_{\bxi} = n^{-1} \sum_{i=1}^n \partial \rho_{\bxi}(y_i, \bx_i) / \partial \bxi$ and $\widehat{\bSigma}_{\bxi} =\Pn \ddot{{\rho}}_{\bxi} = n^{-1} \sum_{i=1}^n \partial^2 \rho_{\bxi}(y_i, \bx_i) / \partial \bxi \partial \bxi^T$. Two important population-level matrices are the information matrix, $\bSigma_{\bxi} = \E ( \widehat{\bSigma}_{\bxi} ) = \E (\Pn \ddot{ {\rho}}_{\bxi})$, and its inverse  $\bTheta_{{\bxi}} = \bSigma_{\bxi}^{-1}$.  With a tuning parameter $\lambda > 0$, the lasso estimator for $\bxi^0$ is defined as
\begin{equation} \label{eq:lasso}
	\widehat{\bxi} = \mathop{\arg\min}_{\bxi = (\beta_0, \, \bbeta^T)^T \in \mR^{p+1}} \left\{ \Pn\rho_{\bxi} + \lambda \| \bbeta \|_1 \right\},
\end{equation}
where we suppress the dependence of
$\lambda$ on $n$ and $p$ for notational ease.
We clarify that we do not penalize the intercept $\beta_0$ in (\ref{eq:lasso}). As such, the theoretical properties for $\widehat{\bxi}$, including the bounds of  estimation errors and prediction errors, are still the same as those in   \citet{van2008high} and \citet{van2014asymptotically}, where all of the parameters are estimated via penalization \citep{buhlmann2011statistics}.

We briefly review the de-biased lasso estimator and its bias decomposition. The first order Taylor expansion of $\Pn \dot{{{\rho}}}_{\bxi^0}$ at $\widehat{\bxi}$ gives
\begin{equation} \label{eq:taylor}
	\Pn \dot{{{\rho}}}_{\bxi^0} = \Pn \dot{{\rho}}_{\widehat{\bxi}} + \Pn \ddot{{\rho}}_{\widehat{\bxi}} (\bxi^0 - \widehat{\bxi}) + {\bDelta},
\end{equation}
where ${\bDelta}$ is a $(p+1)$-dimensional vector of remainder terms with the $j$th element 
\begin{equation*} 
	\Delta_j  = \displaystyle \frac{1}{n} \sum_{i=1}^n  \{  \ddot{\rho}(y_i, a_j^*) - \ddot{\rho}(y_i, \bx_i^T \widehat{\bxi})  \} x_{ij} \bx_i^T (\bxi^0 - \widehat{\bxi}),
\end{equation*}  
in which  $\ddot{\rho}(y,a) = {\partial^2 \rho(y,a)} /{\partial a^2}$, and $a_j^*$ lies between $\bx_i^T\widehat{\bxi}$ and $\bx_i^T\bxi^0$. In linear regression models, ${\bDelta} = \bm{0}$,  which  is not always the case for GLMs. Let  $\bM$ be a $(p+1) \times (p+1)$ matrix approximating $\Thetabeta$. Multiplying both sides of  (\ref{eq:taylor}) by ${\bM}_j$, the $j$th row of $\bM$, we obtain  the following equality for the $j$th component
\begin{equation} \label{eq:derive_bhat_M}
	\widehat{\xi}_j  -  \xi^0_j  + \overbrace{ \left( - \bM_j \Pn \dot{{\rho}}_{\widehat{\bxi}}   \right) }^{I_j} 
	+ \overbrace{\left( -\bM_j {\bDelta} \right) }^{II_j} 
	+ \overbrace{\left( \bM_j \Pn \ddot{{\rho}}_{\widehat{\bxi}} - \bm{e}_j^T \right)  \left( \widehat{\bxi} - \bxi^0 \right)}^{III_j}   = - \bM_j \Pn \dot{{\rho}}_{{\bxi}^0},
\end{equation}
where $\bm{e}_j$ is the unit vector with the $j$th element being 1.
\citet{van2014asymptotically} obtained the above  decomposition by inverting the Karush--Kuhn--Tucker condition while using the node-wise lasso estimate of $\Thetabeta$, denoted by $\widetilde{\bTheta}$, to be the approximation matrix $\bM$. Originally proposed for neighborhood selection in high-dimensional graphs \citep{meinshausen2006high}, the node-wise lasso approach estimates  a sparse matrix $\Thetabeta$ that consists of many zero elements. In (\ref{eq:derive_bhat_M}), the asymptotic bias term $I_j$ is estimable, and $\widehat{\xi}_j + I_j$ corresponds to the de-biased lasso estimator in \citet{van2014asymptotically} with $\bM = \widetilde{\bTheta}$.  {In practice, the  $II_j$ and $III_j$ terms in (\ref{eq:derive_bhat_M}) are not computable because they involve the unknown $\bxi^0$,} and ignoring them  may not help fully remove  biases. Particularly, {the sparse estimator  $\widetilde{{\bTheta}}$ may result in non-negligible $II_j$ and $III_j$ terms compared to $I_j$.}
Consequently, the  $\widetilde{\bTheta}$-based de-biased lasso estimator \citep{van2014asymptotically}  incurs much bias and possesses an unsatisfactory inference performance for GLMs as evidenced by our simulations.


{On the other hand, without the matrix sparsity assumption, one may obtain $\bM$  by solving an optimization problem originally  proposed for linear models \citep{javanmard2014confidence}:
	\begin{equation} \label{eq:qp}
		\min \{ \bzeta^T \widehat{\bSigma}_{\widehat{\bxi}} ~ \bzeta : \bzeta \in \mR^{p+1}, \| \widehat{\bSigma}_{\widehat{\bxi}} ~ \bzeta - \bm{e}_j \|_{\infty} \le \mu_n \}
	\end{equation}
	for $j = 1, \ldots, p+1$ and $\mu_n \ge 0$. Under the conditions in Theorem \ref{thm:main} of Section \ref{sec:theory}, the Hessian matrix $\widehat{\bSigma}_{\widehat{\bxi}}$ is invertible with probability going to one as $n \rightarrow \infty$, and the rows of $\widehat{\bSigma}_{\widehat{\bxi}}^{-1}$ are solutions to (\ref{eq:qp}) when $\mu_n = 0$. 
	As confirmed by our simulations in a variety of  regimes, $\mu_n = 0$  generally performs the best
	in overall bias correction to $II_j + III_j$ and statistical inference
	as $\mu_n$ varies from 0 to 1; see Section \ref{sec:sim}.
	This motivates us to replace $\bM$ with $\widehat{\bTheta} = \widehat{\bSigma}_{\widehat{\bxi}}^{-1}$,  denote by $\widehat{\bTheta}_j$ the $j$th row of $\widehat{\bTheta}$, and reexpress  (\ref{eq:derive_bhat_M}) as} 
\begin{equation} \label{eq:derive_bhat_inv2}
	\widehat{\bxi} - \bxi^0 + \left( - \widehat{\bTheta} \Pn \dot{{\rho}}_{\widehat{\bxi}}   \right) +  \left( -\widehat{\bTheta} {\bDelta} \right)    = - \widehat{\bTheta} \Pn \dot{{\rho}}_{{\bxi}^0}.
\end{equation}
Therefore, we propose a refined de-biased lasso estimator based on  $\widehat{\bTheta}$:
\begin{equation} \label{refined-est}
	\widehat{{\bb}} = \widehat{\bxi} - \widehat{\bTheta} \Pn \dot{{\rho}}_{\widehat{\bxi}}. 
\end{equation}
{We will show that our proposed method possesses desirable asymptotic properties and, in general, performs better than the original de-biased lasso approach \citep{van2014asymptotically} in finite sample settings.}

\section{Theoretical results	\label{sec:theory}}


Without loss of generality, we assume that each covariate has been centered to have mean zero. Let $\bX_{\bxi} = \bW_{\bxi}\bX$ be the weighted design matrix, where $\bW_{\bxi}$ is a diagonal matrix with elements $\omega_i(\bxi) = \{ \ddot{\rho}(y_i, \bx_i^T \bxi) \}^{1/2},~ i = 1,\ldots, n$. Then, for any $\bxi \in \mathbb{R}^{p+1}$,  $\widehat{\bSigma}_{\bxi}$ can be rewritten as $ \widehat{\bSigma}_{\bxi} = \bX_{\bxi}^T \bX_{\bxi}/n$. Recall that the population information matrix $\bSigma_{\bxi} = \E ( \widehat{\bSigma}_{\bxi} ) = \E (\Pn \ddot{ {\rho}}_{\bxi})$, and its inverse matrix is $\bTheta_{{\bxi}} = \bSigma_{\bxi}^{-1}$, {which are respectively equal to $\E (\bX^T \bX /n)$ and $\{\E (\bX^T \bX /n)\}^{-1}$  only for linear models, but not for GLMs}.
The  $\psi_2$-norm  \citep{vershynin2010introduction} is useful for characterizing the convergence rate of $\widehat{\bTheta} = \widehat{\bSigma}_{\widehat{\bxi}}^{-1}$. Explicitly,
for a random variable $Y$, its $\psi_2$-norm  is defined as
$
\| Y \|_{\psi_2} =  \sup_{r \ge 1} r^{-1/2} (\E|Y|^r)^{1/r},
$
and $Y$ is defined to be a sub-Gaussian random variable if 
$\| Y \|_{\psi_2} < \infty$.
For a random vector ${\bZ} \in \mathbb{R}^{p+1} $, its $\psi_2$-norm is defined as $\| {\bZ}\|_{\psi_2} =  \sup_{\| {\ba}\|_2=1} \|\langle {\bZ}, {\ba} \rangle\|_{\psi_2},$ and  ${\bZ}$ is called sub-Gaussian if  $\langle  {\bZ}, \ba  \rangle$ is a  sub-Gaussian random variable for all $\ba \in \mathbb{R}^{p+1}$ with $\|  \ba \|_2 = 1$ \citep{vershynin2010introduction}.   
We list  the regularity conditions as follows.

\begin{assumption} \label{assump1}
	The elements in $\bX$ are bounded almost surely. That is, $\|\bX\|_{\infty} \le K$ almost surely for a constant $K>0$. In addition, the rows of $\bX$ are sub-Gaussian random vectors.
\end{assumption}

\begin{assumption} \label{assump2}
	$\bSigma_{\bxi^0}$ is positive definite with bounded eigenvalues such that, for two positive constants $c_{\mathrm{min}}$ and $c_{\mathrm{max}}$,  $c_{\mathrm{min}} \le \lambdamin(\bSigma_{\bxi^0}) \le \lambdamax(\bSigma_{\bxi^0}) \le c_{\mathrm{max}} < \infty$.  
\end{assumption}

\begin{assumption} \label{assump3}
	The derivatives $\dot{\rho}(y, a) = {\partial \rho(y, a)} / {\partial a}  $ and $\ddot{\rho}(y, a) = {\partial^2 \rho(y, a)} / {\partial a^2} $ exist for all $(y,a)$. Further in some $\delta$-neighborhood, $\delta > 0$, $\ddot{\rho}(y,a)$ is Lipschitz such that for some absolute constant $c_{Lip} > 0$,
	\[
	\displaystyle \max_{a_0 \in \{\bx_i^T \bxi^0\}}   \sup_{\max(|a-a_0|, |\widehat{a} - a_0|) \le \delta}  \sup_{y \in \mathcal{Y}}  \displaystyle \frac{| \ddot{\rho}(y,a) -  \ddot{\rho}(y,\widehat{a})|}{|a-\widehat{a}|}	\le c_{Lip}.
	\]
	And the derivatives are  bounded in the sense that there exist two constants $K_1, K_2 > 0$ such that
	\[
	\begin{array}{c}
		\displaystyle \max_{a_0 \in \{\bx_i^T \bxi^0\}} \sup_{y \in \mathcal{Y}} |\dot{\rho}(y, a_0)| \le K_1, \\
		\displaystyle \max_{a_0 \in \{\bx_i^T \bxi^0\}} \sup_{|a-a_0|\le \delta} \sup_{y \in \mathcal{Y}} |\ddot{\rho}(y,a)| \le K_2. 
	\end{array}
	\]
\end{assumption}

\begin{assumption} \label{assump4}
	$\| \bX \bxi^0 \|_{\infty}$ is bounded from above almost surely.
\end{assumption}

\begin{assumption} \label{assump5}
	The covariance matrix $\E({\bX}^T \bX / n)$ is positive definite with eigenvalues bounded away from 0 and from above.
\end{assumption}

It is common to assume bounded covariates as in Assumption \ref{assump1} and bounded eigenvalues for the information matrix as in Assumption \ref{assump2} in high-dimensional inference literature \citep{van2014asymptotically,ning2017general}. 
Assumption \ref{assump2} 
is needed to derive the rate of convergence for $\widehat{\bxi}$. Assumption \ref{assump3} 
{specifies the required smoothness and local properties
	of the  loss function} $\rho(y, \bx^T \bxi)$ \citep{van2014asymptotically}. 
{Since each element of $\bX \bxi^0$ is the (transformed) conditional mean of $y_i$, it is  reasonable to assume its boundedness in Assumption \ref{assump4} as in  \citet{van2014asymptotically} and \citet{ning2017general} for generalized linear models, and  in \citet{kong2014non} and \citet{fang2017testing} for the Cox models. Also  Assumption  \ref{assump4} is needed to bound the variance of $y_i$ and keep it away from 0 for generalized linear models.}
Assumption \ref{assump5} is a  mild requirement {for  random covariates;  a similar condition on the sample covariance matrix can be found in \citet{wang2011gee}.} Unlike \citet{van2014asymptotically}, we  have  avoided an assumption on the boundedness of $\| \Thetabeta \bx_i \|_{\infty}$, {which is not verifiable and closely related to the sparsity requirement of $\Thetabeta$ under Assumption 1.} 

Let $s_0$ denote the number of non-zero elements in $\bxi^0$, and
consider $\widehat{\bb} = \widehat{\bxi} - \widehat{\bTheta} \Pn \dot{\rho}_{\widehat{\bxi}}$ as defined in \eqref{refined-est}.  Theorem \ref{thm:main} establishes  asymptotic normality for (multiple) linear combinations of  $\widehat{\bb}$,  with a proof provided in Web Appendix A.

\medskip

\begin{theorem} \label{thm:main}
	With $\lambda \asymp \{\log(p)/n \}^{1/2}$, assume that $\displaystyle  {{p^2}/{n}} \rightarrow 0$ and $  s_0 \log(p) ({p/n})^{1/2} \rightarrow 0$ as $n \rightarrow \infty$. {Under  Assumptions 1--5, { we have that $\widehat{\bSigma}_{\widehat{\bxi}}$ is invertible with probability going to one, and that}
		\begin{itemize}
			\item[] (i) for {a constant vector} $\balpha_n \in \mathbb{R}^{p+1}$ with $||\balpha_n||_2 = 1$, 
			\[
			\displaystyle \frac{{n}^{1/2}  \balpha_n^T(\widehat{\bb} - \bxi^0)}{(\balpha_n^T \widehat{\bTheta} \balpha_n)^{1/2}} \overset{\mathcal{D}}{\to} N(0,1) \mathrm{~as~} n \to \infty;
			\]
			\item[] (ii) for a fixed integer $m >  1$ and a constant matrix $\bA_n \in \mR^{m \times (p+1)}$ satisfying $\| \bA_n^T \| \le c_*$ for some constant $c_*$ and $\bA_n \Thetabeta \bA_n^T \to \bF$ for some $\bF \in \mR^{m \times m}$,
			\[
			{n}^{1/2} \bA_n (\widehat{\bb} - \bxi^0) \overset{\mathcal{D}}{\to} N_m (\mathbf{0}, \bF) \mathrm{~as~} n \to \infty.
			\]
		\end{itemize}
	}
\end{theorem}

\begin{remark}
	\normalfont
	Theorem \ref{thm:main} enables us  to construct a $100 \times (1-r)$\% confidence interval for $\balpha_n^T \bxi^0$ as
	$
	[  \balpha_n^T ~ \widehat{\bb} - z_{r/2} ({\balpha_n^T \widehat{\bTheta} \balpha_n / n})^{1/2}  , ~ \balpha_n^T ~ \widehat{\bb} + z_{r/2} ({\balpha_n^T \widehat{\bTheta} \balpha_n / n})^{1/2}  ~ ],
	$
	where $0<r<1$ and  $z_{r/2}$ is the upper $(r/2)$th quantile of the standard normal distribution. Here, $\balpha_n$ can be arbitrarily dense, instead of having only a few non-zero elements such as $\balpha_n = \bm{e}_j$ in \citet{van2014asymptotically}. {A $100 \times (1-r)$\% confidence region for $\bA_n \bxi^0$ can be constructed as $ \{ \ba \in \mR^m: n (\bA_n \widehat{\bb} - \ba)^T (\bA_n \widehat{\bTheta} \bA_n^T)^{-1} (\bA_n \widehat{\bb} - \ba) \le \chi^2_{m, r} \} $, where $\chi^2_{m, r}$ is the upper $r$th quantile of $\chi_m^2$.
	}
\end{remark}


\begin{remark}
	\normalfont
	In a linear regression setting with $\bY= (y_1, \ldots, y_n)^T$, some algebra shows that the proposed estimator (\ref{refined-est}) is identical to the MLE, $(\bX^T \bX)^{-1} \bX^T \bY$, regardless of the choice of the  initial estimate,  $\widehat{\bxi}$. Therefore, as a by-product, Theorem \ref{thm:main} characterizes the asymptotics of the MLE for linear models with  a diverging number of coefficients, {which only requires $p^2/n \to 0$. {This can be shown 
			following a similar proof of Theorem \ref{thm:main} with $\bDelta = 0$  for linear regression models, where $\widehat\bTheta$ is free of regression parameters. It is obvious that regularity conditions can be simplified for linear regression models.}
	} 
	
	
\end{remark}

{
	\begin{remark}
		\normalfont
		Binary covariates, particularly dummy variables for categorical covariates,  satisfy {the assumptions} 
		for Theorem \ref{thm:main}. 
		Therefore, applications of Theorem \ref{thm:main}  encompass inference for categorical covariates, such as drawing inference on comparisons between multiple intervention groups or testing associations of  multi-level categorical covariates with outcomes.
	\end{remark}
}

\section{Numerical experiments \label{sec:sim}}

Under  the ``large $n$, diverging $p$" scenario,  we  compare the estimation biases  and coverage probabilities of confidence intervals across the following  estimators: (i) the original de-biased lasso estimator obtained by using the node-wise lasso estimator $\widetilde{\bTheta}$ in \citet{van2014asymptotically} (ORIG-DS),  (ii) the conventional maximum likelihood estimator (MLE), and (iii) our proposed refined de-biased lasso estimator $\widehat{\bb}$,  based on the inverse matrix estimation  $\widehat{\bTheta} = \widehat{\bSigma}_{\widehat{\bxi}}^{-1}$ (REF-DS).

Simulations using the logistic and Poisson regression models yield  similar observations, and we only report results from logistic regression.  A total of $n=1,000$ observations, each with $p = 40, 100, 300, 400$ covariates, are simulated. Within $\bx_i = (1, \widetilde{\bx}_i^T)^T$, $\widetilde{\bx}_i$ are independently generated from $N_{p} (\bm{0}, \bSigma_x)$ before being truncated at $\pm 6$, and $y_i \mid \bx_i \sim Bernoulli(\mu_i)$, where $\mu_i = \exp(\bx_i^T \bxi^0)/\{ 1+ \exp(\bx_i^T \bxi^0) \}$. The intercept $\beta_0^0 = 0$, and  $\beta_1^0$  varies from 0 to 1.5 with 40 equally spaced increments. In addition, four arbitrarily chosen elements of $\bbeta^0$ take non-zero values, two with $0.5$ and the other two with 1, and  are fixed throughout the simulation. In some settings, the maximum likelihood estimates do not exist due to divergence and  are not shown. The covariance matrix $\bSigma_x$ of $\widetilde{\bx}_i$ takes an autoregressive structure of order 1, i.e. AR(1), with correlation $\rho = 0.7$, or a compound symmetry structure with correlation $\rho=0.7$. The tuning parameter in the $\ell_1$ penalized regression is selected by 10-fold cross-validation, and the tuning parameter for the node-wise lasso estimator $\widetilde{\bTheta}$ is selected using 5-fold cross-validation. Both tuning parameter selection procedures are implemented using  \texttt{glmnet} \citep{friedman2010regularization}. For every $\beta_1^0$ value, we summarize the average bias, empirical coverage probability, empirical standard error and model-based estimated standard error over 200 replications.

Figure~\ref{fig:logit_n1k_ar1} illustrates the simulation results for estimating $\beta_1^0$ under the autoregressive covariance structure, and Figure~\ref{fig:logit_n1k_cs} under the compound symmetry structure. The three methods in comparison behave similarly with only 40 covariates included in the model, and  the MLE yields slightly larger biases. The MLE estimates display much more biases than those obtained  by the other two methods with 100 covariates, and do not always exist due to divergence. When  the MLE estimates do exist, they manifest more variability than the original and refined de-biased lasso estimates,  and are with lower coverage probabilities. {In contrast, our refined de-biased lasso approach  outperforms the MLE because the former utilizes sparse lasso estimates as the initial estimates and is numerically more stable than the latter.}

\begin{figure}
	\centering
	\includegraphics[width=\textwidth]{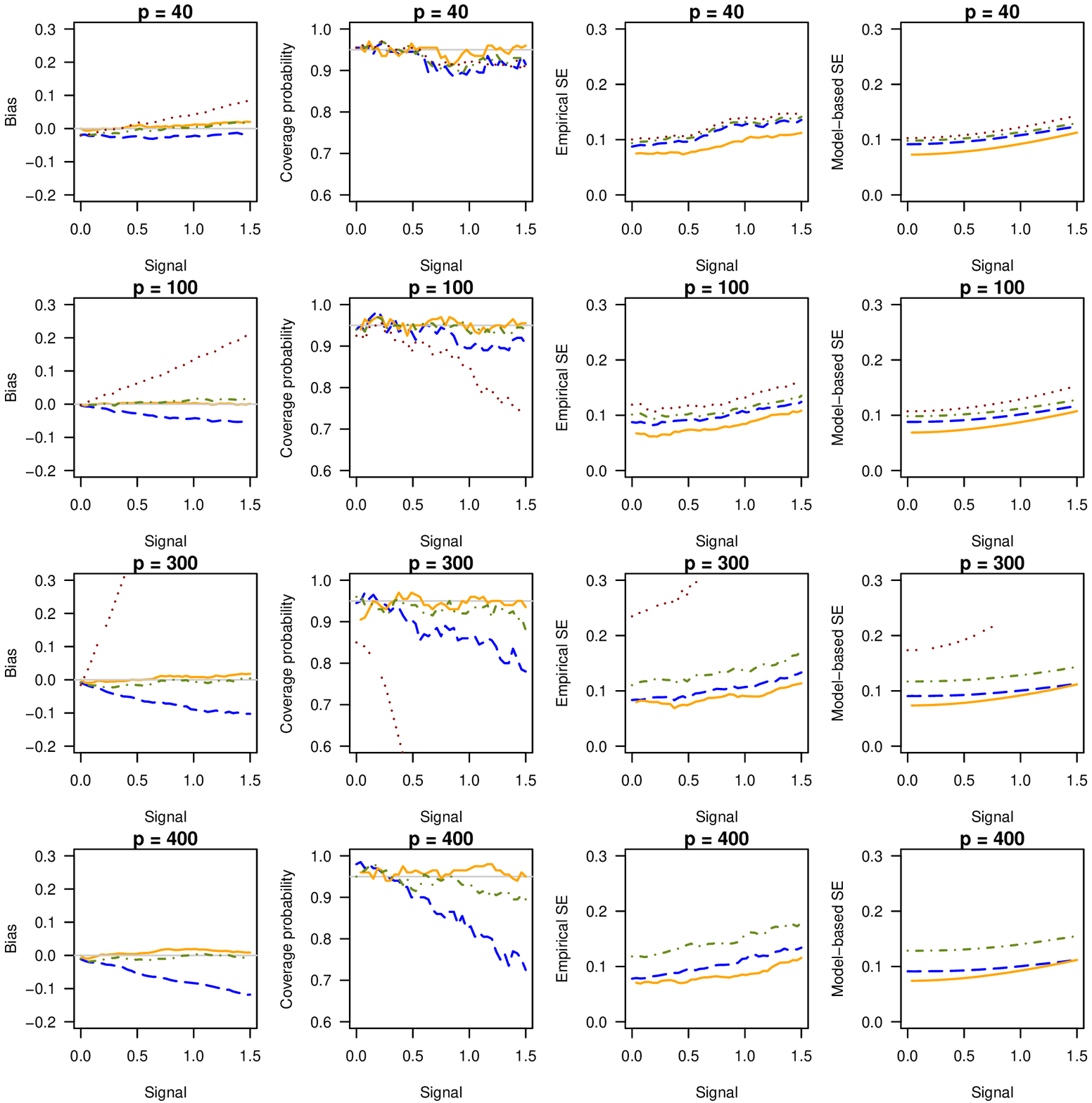}
	\caption{Simulation results: Bias, coverage probability, empirical standard error, and model-based standard error for $\beta_1^0$ in  logistic regression. Covariates are simulated from $N_{p} (\bm{0}, \bSigma_x)$ before being truncated at $\pm 6$, where $\bSigma_x$ has an AR(1) with $\rho=0.7$. The sample size is $n=1,000$ and the number of covariates $p = 40, 100, 300, 400$. The oracle estimator, that is, the maximum likelihood estimator under the true model, is plotted as a reference in orange solid lines. The methods in comparison include our proposed refined de-biased lasso  in olive dot-dash lines, the original de-biased lasso by \citet{van2014asymptotically} in blue dashed lines, and the maximum likelihood estimation in red dotted lines. }
	\label{fig:logit_n1k_ar1}
\end{figure}

\begin{figure}
	\centering
	\includegraphics[width=\textwidth]{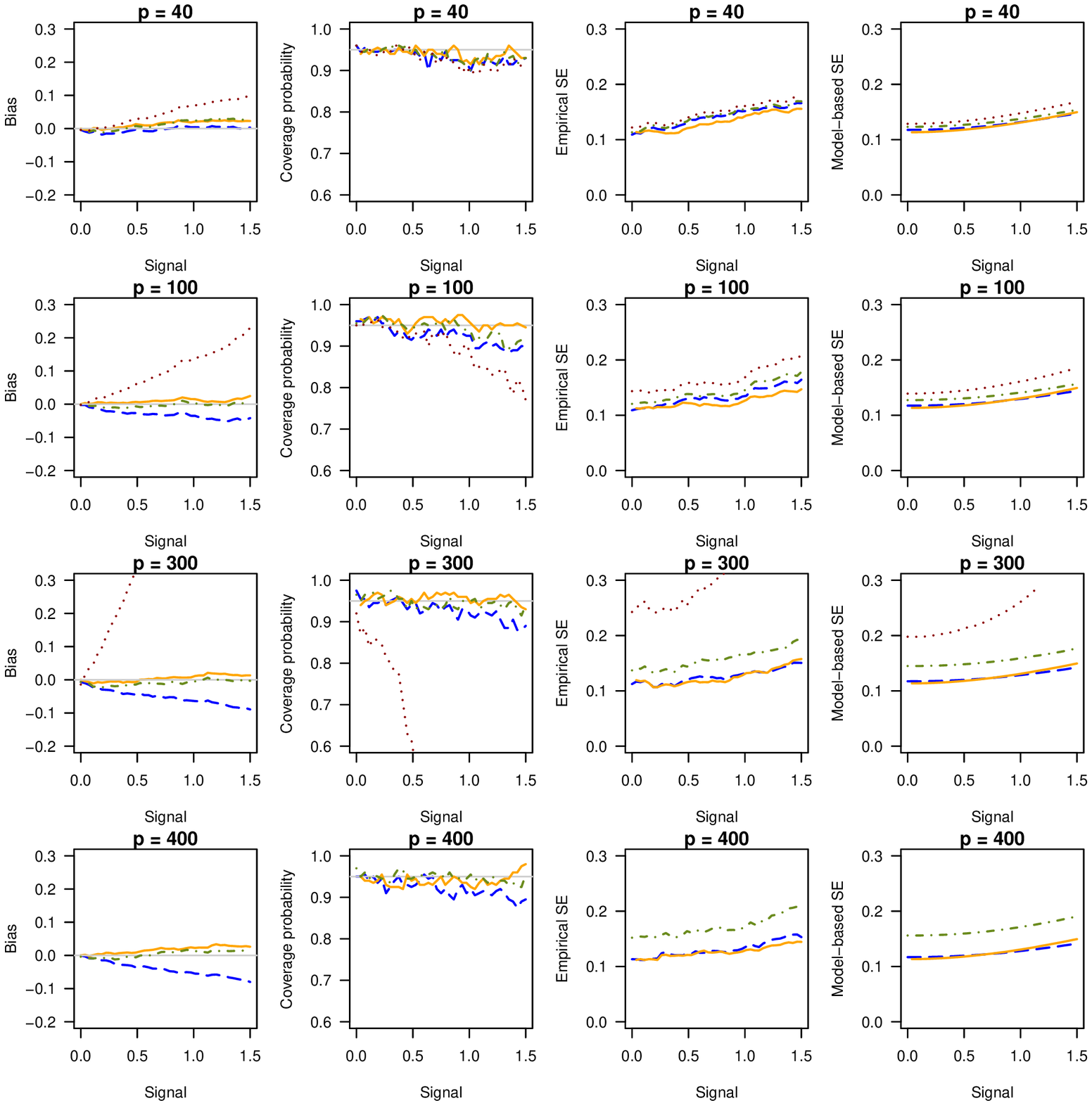}
	\caption{Simulation results: Bias, coverage probability, empirical standard error, and model-based standard error for $\beta_1^0$ in a logistic regression. Covariates are simulated from $N_{p} (\bm{0}, \bSigma_x)$ before being truncated at $\pm 6$, where $\bSigma_x$ has a compound symmetry structure with $\rho=0.7$. The sample size is $n=1,000$ and the number of covariates $p = 40, 100, 300, 400$. The oracle estimator, that is the maximum likelihood estimator under the true model, is plotted as a reference in  orange solid lines. The methods in comparisons include our proposed refined de-biased lasso in olive dot-dash lines, the original de-biased lasso by \citet{van2014asymptotically} in blue dashed lines, and the maximum likelihood estimation in red dotted lines. }
	\label{fig:logit_n1k_cs}
\end{figure}

There are systematic biases in the original de-biased lasso estimator by \citet{van2014asymptotically}, which increase with the magnitude of $\beta_1^0$. 
When  signals are non-zero, the model-based standard errors produced  by \citet{van2014asymptotically} slightly underestimate the true variability. These factors contribute to the poor coverage probabilities of  \citet{van2014asymptotically} when the signal size is not zero. In contrast, the refined de-biased lasso estimator gives the smallest biases and has an empirical coverage probability closest to the nominal level across different settings, though with slightly higher variability than \citet{van2014asymptotically}.
This is likely because our proposed de-biased lasso approach does not utilize a penalized estimator of the inverse information matrix.
{We take note that as the refined  de-biased lasso  method needs to invert the Hessian matrix, 
	which could become more ill-conditioned if the dimension increases, its  performance may deteriorate as the dimension of covariates increases.}

As we alluded to in Section \ref{sec:method}, the refined de-biased lasso estimator is  related to \citet{javanmard2014confidence}, and  we have conducted additional simulations to compare them, referred to as ``REF-DS" and ``Tuning" respectively. 
Figure \ref{fig:tuning_large}, which
depicts the results of a logistic regression model with $n=500$ observations and $p=40, 100, 200, 300, 400$ covariates,  shows that 
$\mu_n = 0$ generally performs the best in bias correction and honest confidence interval coverage when $\mu_n$ varies from 0 to 1; see the simulation setup and additional results in Web Appendix B.


\begin{figure}
	\centering
	\includegraphics[width=0.8\textwidth]{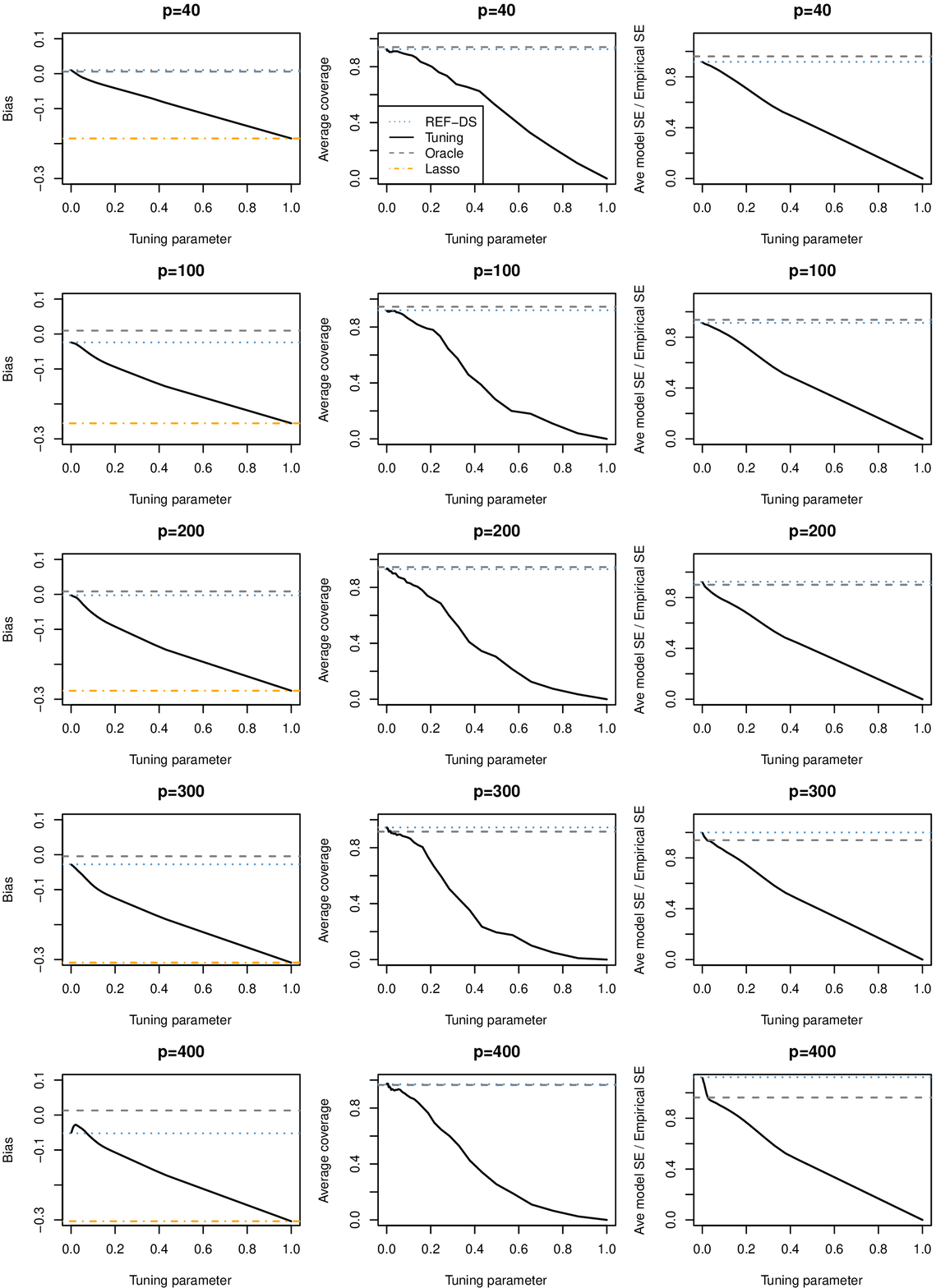}
	\caption{{Simulation results: Bias, coverage probability, ratio between average model-based standard error and empirical standard error in a logistic regression to verify the selection of the tuning parameter $\mu_n=0$ in Eq. \eqref{eq:qp} for $\xi_j^0=1$}. }
	\label{fig:tuning_large}
\end{figure}

\section{Boston lung cancer data analysis \label{sec:app}}

Lung cancer is the top cause of cancer death in the United States. The Boston Lung Cancer Survival Cohort (BLCSC), one of the largest  hospital-based cohorts in the country,  investigates the molecular causes of lung cancer.  Recruited to the study were the lung cancer cases and controls from the Massachusetts General Hospital and the Dana-Farber Cancer Institute since 1992 \citep{miller2002combinations}.
We apply the proposed refined de-biased lasso approach, together with the method by \citet{van2014asymptotically} and the MLE for comparison, to a subset of the BLCSC data and  examine the joint effects of SNPs from nine target genes on the overall risk of lung cancer. 

Genotypes from Axiom array and clinical information were originally available on 1,459 individuals. Out of those individuals, 14 (0.96\%) had missing smoking status, 8 (0.55\%) had missing race information, and 1,386 (95\%) were Caucasian. We include a final number of $n=1,374$ Caucasians with complete data, where $n_0=723$ were controls and $n_1=651$ were cases. Denote the binary disease outcome by $y_i = 1$ for cases and $0$ for controls. Among the 1,077 smokers, 595 had lung cancer, whereas  out of the 297 non-smokers, 56 were cases. Other demographic characteristics, such as education level, gender and age, are summarized in Web Appendix C. Using the target gene approach, we focus on the following lung cancer-related genes: \textit{AK5} on region 1p31.1, \textit{RNASET2} on region 6q27,  \textit{CHRNA2} and \textit{EPHX2} on region 8p21.2, \textit{BRCA2} on region 13q13.1,  \textit{SEMA6D} and \textit{SECISBP2L} on region 15q21.1,  \textit{CHRNA5} on region 15q25.1, and \textit{CYP2A6} on region 19q13.2. These genes may harbor SNPs  associated  with the overall lung cancer risks \citep{mckay2017large}. In our dataset, each SNP is coded as 0,1,2, reflecting the number of copies of the minor allele, and minor alleles are assumed  to have additive effects. After applying filters on the minor allele frequency,  genotype call rate, and excluding highly correlated SNPs, 103 SNPs remain in the model.  Since smoking may modify associations  between lung cancer risks and SNPs, for example, those residing in region 15q25.1  \citep{gabrielsen2013association,amos2008genome}, we conduct analysis stratified by smoking status. Among the smokers and non-smokers, we fit separate logistic regression models, adjusting for education, gender and age. 

We apply  these methods to  draw inference on all of the 107 predictors, two of which are dummy variables for education originally with three levels, no high school, high school and at least 1-2 years of college. Our data analysis may shed light on the molecular mechanism underlying lung cancer. Due to limited space, Table~\ref{tab:smoker} 
lists the estimates for 11 selected SNPs and demographic variables among smokers, and Table~\ref{tab:nonsmoker} for non-smokers. These SNPs are listed as they are significant based on at least  one of the three methods among either the smokers or the non-smokers.  Details of the other SNPs are omitted.
Since the number of the non-smokers is only about one third of the smokers, the MLE has the largest standard errors and tends to break down among the non-smokers; see, for example, AX-62479186 in Table \ref{tab:nonsmoker},
whereas the two de-biased lasso methods give more stable estimates. The estimates by our proposed refined de-biased lasso method (REF-DS) and the method by \citet{van2014asymptotically} (ORIG-DS) share more similarities in the smokers in Table~\ref{tab:smoker} than in the non-smokers in Table~\ref{tab:nonsmoker}.  Overall, the method by \citet{van2014asymptotically} has slightly narrower confidence intervals than our proposed de-biased lasso estimator due to the penalized estimation for ${\bTheta}_{\bxi^0}$.   These results generally agree with our  simulation studies. 

For some SNPs,  our proposed method and the method by \citet{van2014asymptotically} yield estimates with  opposite directions; see  AX-38419741 and AX-15934253 in Table~\ref{tab:smoker},
and AX-42391645 in Table~\ref{tab:nonsmoker}.
Among the non-smokers, the 95\% confidence interval for AX-31620127 in \textit{SEMA6D} by our proposed method is all positive and excludes 0, while the confidence interval by the method of \citet{van2014asymptotically} includes 0; the directions for AX-88907114 in \textit{CYP2A6} are the opposite in Table~\ref{tab:nonsmoker}.
\textit{CHRNA5} is a gene  known  for predisposition to nicotine dependence \citep{hallden2016gene,amos2008genome,gabrielsen2013association}. 
Though AX-39952685 and  AX-88891100 in \textit{CHRNA5} are not significant at level 0.05 in marginal analysis among the smokers, their 95\% confidence intervals in Table~\ref{tab:smoker}
exclude 0 by all of the three methods.  Indeed, AX-88891100, or rs503464, mapped to the same physical location in the dbSNP database, 
was found to ``decrease \textit{CHRNA5} promoter-derived luciferase activity" \citep{doyle2011vitro}. The same SNP was also reported to be significantly associated with nicotine dependence at baseline, as well as response to varenicline, bupropion, nicotine replacement therapy for smoking cessation \citep{pintarelli2017pharmacogenetic}.  AX-39952685 was found to be strongly correlated with SNP AX-39952697 in \textit{CHRNA5}, which was mapped to the same physical location as rs11633585 in dbSNP. All of these markers  were found to be significantly associated with nicotine dependence \citep{stevens2008nicotinic}. The stratified analysis also suggests molecular mechanisms of lung cancer differ between smokers and non-smokers, though additional confirmatory studies are needed. 

\begin{sidewaystable}
	\centering
	\caption{Estimated coefficients for demographic variables and eleven SNPs in a logistic regression model among the smokers. The other SNPs are omitted from the table} 		\label{tab:smoker}
	\small
	\begin{tabular}{cccccccccccc}
		\hline
		&  & &   \multicolumn{3}{c}{\textit{REF-DS}\footnote{The proposed refined de-biased lasso based on the inverted Hessian matrix}} & \multicolumn{3}{c}{\textit{ORIG-DS}\footnote{The original de-biased lasso based on the node-wise lasso estimator for $\Thetabeta$ by \citet{van2014asymptotically}}} & \multicolumn{3}{c}{\textit{MLE}\footnote{The maximum likelihood estimation approach}}  \\
		\multicolumn{3}{r}{Demographic variable} & Est\footnote{The point estimate for each coefficient} & SE\footnote{The model-based standard error} & 95\% CI\footnote{Confidence interval}  & Est & SE & 95\% CI & Est & SE & 95\% CI   \\ 
		\hline
		\multicolumn{3}{r}{Education: No high school} & 0.44 & 0.20 & (0.05, 0.83) & 0.48 & 0.19 & (0.11, 0.85) & 0.52 & 0.22 & (0.09, 0.94) \\ 
		\multicolumn{3}{r}{Education: High school graduate} & 0.00 & 0.15 & (-0.29, 0.28) & -0.02 & 0.14 & (-0.29, 0.25) & -0.01 & 0.16 & (-0.32, 0.30) \\ 
		\multicolumn{3}{r}{Gender: Male} & -0.14 & 0.13 & (-0.40, 0.12) & -0.16 & 0.13 & (-0.40, 0.09) & -0.15 & 0.14 & (-0.43, 0.13) \\ 
		\multicolumn{3}{r}{Age in years} & 0.04 & 0.07 & (-0.09, 0.17) & 0.05 & 0.06 & (-0.07, 0.17) & 0.05 & 0.07 & (-0.09, 0.19) \\ 
		& & & & & & & & & & & \\ 
		\hline
		SNP & Pos\footnote{The physical position of a SNP on a chromosome based on Assembly GRCh37/hg19} & Gene & Est & SE & 95\% CI  & Est & SE & 95\% CI & Est & SE & 95\% CI   \\ 
		AX-15319183 & 6:167352075 &  \textit{RNASET2} & 0.01 & 0.19 & (-0.36, 0.39) & -0.03 & 0.18 & (-0.39, 0.33) & 0.02 & 0.20 & (-0.38, 0.42) \\ 
		AX-41911849 & 6:167360724 &  \textit{RNASET2} & 0.43 & 0.22 & (0.00, 0.86) & 0.44 & 0.20 & (0.06, 0.83) & 0.49 & 0.24 & (0.03, 0.96) \\ 
		AX-42391645 & 8:27319769 &  \textit{CHRNA2} & 0.01 & 0.16 & (-0.29, 0.32) & -0.01 & 0.14 & (-0.28, 0.26) & 0.01 & 0.16 & (-0.31, 0.34) \\ 
		\textbf{AX-38419741} & 8:27319847  & \textit{CHRNA2} & \textbf{0.11} & 0.35 & (-0.59, 0.80) & \textbf{-0.14} & 0.31 & (-0.75, 0.48) & 0.13 & 0.37 & (-0.60, 0.86) \\ 
		\textbf{AX-15934253} & 8:27334098  & \textit{CHRNA2} & \textbf{-0.15} & 0.44 & (-1.02, 0.71) & \textbf{0.06} & 0.39 & (-0.70, 0.82) & -0.19 & 0.47 & (-1.12, 0.74) \\ 
		AX-12672764 & 13:32927894 & \textit{BRCA2} & -0.07 & 0.19 & (-0.44, 0.31) & -0.10 & 0.16 & (-0.40, 0.21) & -0.07 & 0.20 & (-0.47, 0.33) \\ 
		AX-31620127 & 15:48016563  & \textit{SEMA6D} & 0.79 & 0.26 & (0.28, 1.31) & 0.79 & 0.25 & (0.30, 1.28) & 0.96 & 0.30 & (0.37, 1.55) \\ 
		\textbf{AX-88891100} & 15:78857896  & \textit{CHRNA5} & 0.87 & 0.36 & (0.17, 1.57) & 0.79 & 0.32 & (0.16, 1.41) & 0.98 & 0.39 & (0.22, 1.74) \\ 
		\textbf{AX-39952685} & 15:78867042  & \textit{CHRNA5} & 0.99 & 0.47 & (0.07, 1.91) & 0.82 & 0.38 & (0.09, 1.56) & 1.11 & 0.50 & (0.13, 2.08) \\ 
		AX-62479186 & 15:78878565  & \textit{CHRNA5} & 0.41 & 0.41 & (-0.40, 1.22) & 0.46 & 0.39 & (-0.31, 1.23) & 0.46 & 0.45 & (-0.41, 1.33) \\ 
		AX-88907114 & 19:41353727  & \textit{CYP2A6} & 0.52 & 0.34 & (-0.16, 1.19) & 0.49 & 0.33 & (-0.15, 1.13) & 0.58 & 0.38 & (-0.15, 1.32) \\ 
		$\vdots$ & & & & & & & & & & &  \\
		\hline
	\end{tabular}
\end{sidewaystable}

\begin{sidewaystable}
	\centering
	\caption{Estimated coefficients for demographic variables and eleven SNPs in a logistic regression model among the non-smokers. The other SNPs are omitted from the table} 		\label{tab:nonsmoker}
	\small
	\begin{tabular}{cccccccccccc}
		\hline
		&  & &   \multicolumn{3}{c}{\textit{REF-DS}\footnote{The proposed refined de-biased lasso based on the inverted Hessian matrix}} & \multicolumn{3}{c}{\textit{ORIG-DS}\footnote{The original de-biased lasso based on the node-wise lasso estimator for $\Thetabeta$ by \citet{van2014asymptotically}}} & \multicolumn{3}{c}{\textit{MLE}\footnote{The maximum likelihood estimation approach}}  \\
		\multicolumn{3}{r}{Demographic variable} & Est\footnote{The point estimate for each coefficient} & SE\footnote{The model-based standard error} & 95\% CI\footnote{Confidence interval}  & Est & SE & 95\% CI & Est & SE & 95\% CI   \\ 
		\hline
		\multicolumn{3}{r}{Education: No high school} & -0.84 & 0.93 & (-2.67, 0.99) & -0.58 & 0.78 & (-2.11, 0.95) & -10.53 & 3.82 & (-18.01, -3.04) \\ 
		\multicolumn{3}{r}{Education: High school graduate} & -1.68 & 0.52 & (-2.69, -0.66) & -1.56 & 0.43 & (-2.39, -0.72) & -11.23 & 3.75 & (-18.58, -3.88) \\ 
		\multicolumn{3}{r}{Gender: Male} & -0.30 & 0.41 & (-1.10, 0.51) & -0.16 & 0.32 & (-0.78, 0.46) & -1.94 & 1.03 & (-3.96, 0.09) \\ 
		\multicolumn{3}{r}{Age in years} & -0.52 & 0.20 & (-0.91, -0.13) & -0.56 & 0.16 & (-0.87, -0.26) & -2.59 & 0.98 & (-4.51, -0.67) \\ 
		&  &  &  &  &  &  &  &  &  &  &  \\ \hline
		SNP & Pos\footnote{The physical position of a SNP on a chromosome based on Assembly GRCh37/hg19}  & Gene & Est & SE & 95\% CI  & Est & SE & 95\% CI & Est & SE & 95\% CI  \\
		AX-15319183 & 6:167352075  & \textit{RNASET2} & -0.71 & 0.55 & (-1.78, 0.36) & 0.01 & 0.40 & (-0.79, 0.80) & -4.32 & 1.84 & (-7.92, -0.71) \\ 
		AX-41911849 & 6:167360724  & \textit{RNASET2} & 0.69 & 0.65 & (-0.59, 1.97) & 0.37 & 0.47 & (-0.55, 1.29) & 4.46 & 2.00 & (0.54, 8.39) \\ 
		\textbf{AX-42391645} & 8:27319769  & \textit{CHRNA2} & \textbf{-0.11} & 0.49 & (-1.07, 0.85) & \textbf{0.18} & 0.30 & (-0.41, 0.78) & -1.90 & 2.00 & (-5.81, 2.02) \\ 
		AX-38419741 & 8:27319847  & \textit{CHRNA2} & 0.50 & 1.04 & (-1.54, 2.53) & 0.23 & 0.61 & (-0.97, 1.42) & 3.37 & 3.00 & (-2.51, 9.26) \\ 
		AX-15934253 & 8:27334098  & \textit{CHRNA2} & 0.11 & 1.40 & (-2.64, 2.86) & 0.38 & 0.82 & (-1.23, 1.98) & 5.37 & 4.21 & (-2.88, 13.62) \\ 
		AX-12672764 & 13:32927894  & \textit{BRCA2} & -0.83 & 0.62 & (-2.04, 0.37) & -0.57 & 0.38 & (-1.32, 0.18) & -8.25 & 2.64 & (-13.42, -3.08) \\ 
		\textbf{AX-31620127} & 15:48016563  & \textit{SEMA6D} & 1.77 & 0.75 & \textbf{(0.30, 3.24)} & 0.43 & 0.46 & \textbf{(-0.48, 1.34)} & 9.23 & 3.27 & (2.81, 15.64) \\ 
		AX-88891100 & 15:78857896  & \textit{CHRNA5} & 0.78 & 1.18 & (-1.54, 3.10) & 1.15 & 0.87 & (-0.56, 2.85) & 1.54 & 3.17 & (-4.68, 7.75) \\ 
		AX-39952685 & 15:78867042  & \textit{CHRNA5} & -0.54 & 1.30 & (-3.09, 2.01) & -0.99 & 0.73 & (-2.41, 0.44) & -2.85 & 3.98 & (-10.65, 4.96) \\ 
		\textbf{AX-62479186} & 15:78878565  & \textit{CHRNA5} & -1.28 & 1.34 & (-3.92, 1.35) & -1.33 & 1.10 & (-3.49, 0.82) & \textbf{-19.64} & \textbf{3410.98} & (-6705.04, 6665.75) \\ 
		\textbf{AX-88907114} & 19:41353727  & \textit{CYP2A6} & 0.86 & 0.88 & \textbf{(-0.86, 2.59)} & 1.40 & 0.68 & \textbf{(0.06, 2.74)} & 3.52 & 2.18 & (-0.75, 7.78) \\ 
		$\vdots$ & & & & & & & & & & &  \\
		\hline
	\end{tabular}
\end{sidewaystable}

\section{Concluding remarks \label{sec:discuss}}

We have proposed a refined de-biased lasso estimating method for GLMs by directly   inverting Hessian matrices in the ``large $n$, diverging $p$" framework. We have showed that if $p^2/n = o(1)$ and $({p/n})^{1/2} s_0 \log(p) = o(1)$, along with some other mild conditions,  any linear combinations of the resulting estimates are asymptotically normal and can be used for constructing hypothesis tests and confidence intervals.  By way of empirical studies, we have showed that when $p$ is  small relative to $n$, the proposed refined de-biased lasso yields estimates nearly identical to the MLE and the original de-biased lasso by \citet{van2014asymptotically}. In contrast,  the proposed method outperforms the latter two in bias correction and confidence interval coverage probabilities when $p < n$ but  $p$ is relatively large, indicating a broad applicability. 

{Additional simulations for linear regression models (results not shown here) indicate that, however, both our proposed method (equivalent to the MLE, see Remark 2) and the original de-biased lasso method perform well with no obvious difference between these two methods for wide ranges of $p/n$. This is likely due to the fact that the Hessian matrix for a linear model is free of regression parameters.}

{Theorem \ref{thm:main} gives some sufficient range of $p$ relative to $n$ to guide practical settings, but does not necessarily exhaust all possible working scenarios in a finite sample setting. In fact, we have shown through simulations that the asymptotic approximations given in Theorem \ref{thm:main} work well in finite sample settings with wide ranges of $p$ and $n$.} Nevertheless, searching for more relaxed conditions of $p$ and $n$ warrants more in-depth investigations.

{With a slightly stronger requirement of $s_0 \log(p) (p/n)^{1/2} \to 0$ than $s_0 \log(p) /\sqrt{n} \to 0$ specified in \citet{van2014asymptotically}, Theorem 1 obtains stronger results than theirs in i)  drawing  inference for any linear combinations of regression coefficients, ii) releasing sparsity assumptions on $\Thetabeta$, and iii) dropping the  boundedness assumption on $\| {\bTheta}_{\bxi^0} \bx_i \|_{\infty}$; 
	see Web Appendix D for  detailed discussion. {Moreover, a referee pointed out a recent work on linear regression models \citep{bellec2018slope}  that may help provide slightly less stringent sparsity conditions by relaxing the logarithmic factor; however, such generalization to GLMs is beyond our scope.}
}


Lastly, we comment on the difficulties of applying some existing methods to draw inference with high-dimensional GLMs. With extensive simulations, we have discovered unsatisfactory bias correction and confidence interval coverage with the original de-biased lasso in GLM settings \citep{van2014asymptotically}; for example, see  the simulation results under the ``large $p$, small $n$" scenario in Web Appendix B.  Our further investigation pinpoints an essential assumption that hardly holds for GLMs in general, which is that the number of non-zero elements in the rows of the high-dimensional inverse information matrix $\bTheta_{\bxi^0}$ is sparse and of order $o [ \{ n/\log(p) \}^{1/2} ]$  \citep{van2014asymptotically}. The theoretical developments in \citet{van2014asymptotically} rely heavily on this sparse matrix assumption. The $\ell_0$ sparsity conditions on high-dimensional matrices are not uncommon in the literature of high-dimensional inference. A related $\ell_0$ sparsity condition on $\bm{w}^* = {\bI^{* -1}}_{\mathbf{\gamma}\mathbf{\gamma}} \bI^*_{\mathbf{\gamma} \theta}$ can be found in \citet{ning2017general}, where $\bI^*$ is the information matrix under the truth, but it is not well justified in a general setting for GLMs. When testing a global null hypothesis $\bbeta^0= {0}$,  the sparsity of $\bTheta_{\bxi^0}$ reduces to the sparsity of the covariate precision matrix, which becomes less of an issue \citep{ma2020global}. Therefore, we generally do not recommend the de-biased lasso method  when $p>n$ for GLMs.


\section*{Acknowledgements}
This work was partially supported by the United States National Institutes of Health (R01AG056764, R01CA249096 and U01CA209414), and the National Science Foundation (DMS-1915711). The authors thank Dr. David Christiani for providing the Boston Lung Cancer Survival Cohort data.

\section*{Data Availability Statement}
The Boston Lung Cancer Survival Cohort data are not publicly available due to access restrictions. 

\section*{Supporting Information}
Web Appendices referenced in Sections~\ref{sec:theory}, \ref{sec:sim}, \ref{sec:app} and \ref{sec:discuss} are available at the end of this article.  R code and a simulated example are available at \url{https://github.com/luxia-bios/DebiasedLassoGLMs/}.
\vspace*{-8pt}



\bibliographystyle{apalike}
\bibliography{paper-ref}


\newpage

\setcounter{equation}{0}
\renewcommand{\theequation}{S\arabic{equation}}
\setcounter{theorem}{0}
\renewcommand{\thetheorem}{S\arabic{theorem}}
\setcounter{section}{0}
\renewcommand{\thesection}{Web Appendix \Alph{section}}
\setcounter{figure}{0}
\renewcommand{\thefigure}{S\arabic{figure}}
\setcounter{table}{0}
\renewcommand{\thetable}{S\arabic{table}}

\begin{center}
	\Large \bf Supporting Information for ``De-biased Lasso for Generalized Linear Models with A Diverging Number of Covariates"
\end{center}

{
\centering
	\author{Lu Xia\textsuperscript{1}, Bin Nan\textsuperscript{2*}, and Yi Li\textsuperscript{3*}\\
		\small 
		\textsuperscript{1}Department of Biostatistics, University of Washington, Seattle, WA, {xialu@uw.edu}  \\
		\small
		\textsuperscript{2}Department of Statistics, University of Californina, Irvine, CA, {nanb@uci.edu} \\
		\small
		\textsuperscript{3}Department of Biostatistics, University of Michigan, Ann Arbor, MI, {yili@umich.edu} \\
		\small 
		*To whom correspondence should be addressed \\
	}
}

\begin{abstract}
	We present the technical proofs of Theorem 1 in the main text and  the lemmas used for the  theorem in Web Appendix A. Additional simulation results and descriptive statistics for the Boston Lung Cancer Survivor Cohort are provided in Web Appendix B and Web Appendix C, respectively. Web Appendix D entails additional discussion on the difference between sparsity assumptions in our work and \citet{van2014asymptotically}.
\end{abstract}


\section{Technical proofs}

\subsection{ Lemmas}

We list three lemmas that are used for proving Theorem 1.
Without loss of generality, we denote the dimension of the parameter $\bxi$ by $p$ instead of $(p+1)$ to simplify the notation in the  proofs. Consequently, the matrices such as $\bSigma_{\bxi}$ and $\bTheta_{\bxi}$ are considered as $p\times p$ matrices. This simplification of notation does not affect the following derivations.

Lemma \ref{lemma1}   bounds  the estimation error and the prediction error of the lasso estimator under our assumptions. 

\begin{lemma} \label{lemma1}
	Under Assumptions 1--5, we have $\| \widehat{\bxi} - \bxi^0 \|_1 = \OP(s_0\lambda)$ and $\| \bX (\widehat{\bxi} - \bxi^0) \|_2^2 /n = \OP(s_0 \lambda^2)$.
\end{lemma}

\begin{proof}
	
	Because $\lambda_{\mathrm{min}}(\bSigma_{\bxi^0}) > 0$ in Assumption 2, the compatibility condition holds for all index sets $S \subset \{ 1, \ldots, p\}$ by Lemma 6.23 \citep{buhlmann2011statistics} and the fact that the adaptive restricted eigenvalue condition implies the compatibility condition. Exploiting  Hoeffding's concentration inequality, we have $\| \widehat{\bSigma}_{\bxi^0} - \bSigma_{\bxi^0} \|_{\infty} = \OP[ \{log(p)/n\}^{1/2} ]$. Then by Lemma 6.17 of \citet{buhlmann2011statistics}, we have the $\widehat{\bSigma}_{\bxi^0}$-compatibility condition. Finally, the first part of Lemma \ref{lemma1} follows from Theorem 6.4 in \citet{buhlmann2011statistics}.  
	
	For the second claim, \citet{ning2017general} showed that $$(\widehat{\bxi} - \bxi^0)^T \widehat{\bSigma}_{\bxi^0} (\widehat{\bxi} - \bxi^0) = (\widehat{\bxi} - \bxi^0)^T (X^T W_{\xi^0}^2 X/n) (\widehat{\bxi} - \bxi^0)  = \OP(s_0 \lambda^2),$$ then under Assumption 4, {the variance terms in $W_{\xi^0}^2$ are bounded away from 0, and} we obtain the desired result that $\| \bX (\widehat{\bxi} - \bxi^0) \|_2^2 /n = \OP(s_0 \lambda^2)$. 
\end{proof}

Lemma \ref{lemma2} depicts the convergence rate of the inverse Hessian matrix $\hTheta$ to the true inverse information matrix $\Thetabeta$.

\begin{lemma} \label{lemma2}
	Suppose the covariate vectors $x_i, ~ i = 1, \ldots, n$, are independent and identically distributed sub-Gaussian random vectors. Under Assumptions 1--5, if we further assume that $s_0 \lambda \rightarrow 0$ and  $p/n \rightarrow 0$, then $\hTheta$ converges to $\Thetabeta$ such that
	\[
	\| \hTheta - \bTheta_{\bxi^0} \| = \OP \{  ({{p} / {n}})^{1/2} + s_0 \lambda  \}.
	\]
\end{lemma}

\begin{proof}
	Since $\widehat{\bSigma}_{\widehat{\bxi}}^{-1} - \bSigma_{\bxi^0}^{-1} = \widehat{\bSigma}_{\widehat{\bxi}}^{-1} \left(\bSigma_{\bxi^0} -  \widehat{\bSigma}_{\widehat{\bxi}}  \right) \bSigma_{\bxi^0}^{-1}$, we have 
	\begin{equation} \label{eq:inv}
		\| \widehat{\bSigma}_{\widehat{\bxi}}^{-1} - \bSigma_{\bxi^0}^{-1} \| \le  \| \widehat{\bSigma}_{\widehat{\bxi}}^{-1}  \| \cdot \| \widehat{\bSigma}_{\widehat{\bxi}}  - \bSigma_{\bxi^0}  \|  \cdot \| \bSigma_{\bxi^0}^{-1} \|.
	\end{equation}
	By Assumption 2, $\| \bSigma_{\bxi^0}^{-1} \|$ is bounded. We obtain the convergence rate of $\| \widehat{\bSigma}_{\widehat{\bxi}}^{-1} - \bSigma_{\bxi^0}^{-1} \|$ by calculating the convergence rate of $\| \widehat{\bSigma}_{\widehat{\bxi}} - \bSigma_{\bxi^0}  \|$ and showing that $\| \widehat{\bSigma}_{\widehat{\bxi}}^{-1}  \| $ is bounded with probability going to 1.
	
	Note that $\| \widehat{\bSigma}_{\widehat{\bxi}} - \bSigma_{\bxi^0}  \| \le \| \widehat{\bSigma}_{\widehat{\bxi}} - \widehat{\bSigma}_{\bxi^0} \| + \| \widehat{\bSigma}_{\bxi^0} -  \bSigma_{\bxi^0}  \|$. When the rows of $\bX$ are sub-Gaussian, so are the rows of $\bX_{\bxi^0}$ due to the boundedness of the weights $w_i$ in Assumption 3. It  can be shown that $L = \| \bSigma_{\bxi^0}^{-1/2} \bx_1 \omega_1(\bxi^0) \|_{\psi_2} = \mathcal{O}(1)$. First, for $\| \widehat{\bSigma}_{\bxi^0} -  \bSigma_{\bxi^0}  \|$, \citet{vershynin2010introduction} shows that for every $t > 0$, it holds with probability at least  $1 - 2\exp(-c_L^{\prime} t^2)$ that
	\begin{equation}
		\| \widehat{\bSigma}_{\bxi^0} -  \bSigma_{\bxi^0} \| \le \| \bSigma_{\bxi^0}  \| \max(\delta, \delta^2) \le c_{\mathrm{max}} \max(\delta, \delta^2),
	\end{equation}
	where $\delta = C_L (p/n)^{1/2} + {t}/{n}^{1/2}$. Here $C_L$, $c_L^{\prime} > 0$ depend only on  $L = \| \bSigma_{\bxi^0}^{-{1}/{2}} \bx_1 \omega_1(\bxi^0) \|_{\psi_2}$. In fact $c_L^{\prime} = c_1/L^4$ and $C_L = L^2 ({\log 9 / c_1})^{1/2}$, where $c_1$ is an absolute constant. For $s > 0$ and $t = s C_L {p}^{1/2}$, the probability becomes $1 - 2\exp(-c_2 s^2 p)$, $c_2 > 0$ being some absolute constant, and $\delta = (s+1) C_L (p/n)^{1/2}$. Thus $\| \widehat{\bSigma}_{\bxi^0} -  \bSigma_{\bxi^0} \| = \mathcal{O}_P \{  L^2 ({p}/{n})^{1/2} \} = \mathcal{O}_P \{  ({p}/{n})^{1/2} \} $. 
	
	Note that 
	\[
	\begin{array}{rcl}
		\| \widehat{\bSigma}_{\widehat{\bxi}} - \widehat{\bSigma}_{\bxi^0} \| & = & \| {\bX}^T (\bW^2_{\widehat{\bxi}} - \bW^2_{\bxi^0}) \bX / n \|  \\ 
		& \le & \| {\bX}^T \| \cdot \| \bX \| / n \cdot \| \bW^2_{\widehat{\bxi}} - \bW^2_{\bxi^0} \|  \\
		& = & \lambdamax({\bX}^T \bX /n) \cdot  \| \bW^2_{\widehat{\bxi}} - \bW^2_{\bxi^0} \|.
	\end{array}
	\]
	By Assumptions 1 and  3, \begin{equation}
		\begin{array}{rcl}
			\| \bW^2_{\widehat{\bxi}} - \bW^2_{\bxi^0} \| & = & \max_i | \ddot{\rho}(y_i , \bx_i^T \widehat{\bxi}) - \ddot{\rho}(y_i , \bx_i^T \bxi^0)  | \\
			& \le &  c_{Lip}  \cdot \max_i | \bx_i^T (\widehat{\bxi} - \bxi^0)| \\
			& \le &  c_{Lip} K  \cdot \| \widehat{\bxi} - \bxi^0 \|_1.
		\end{array}
	\end{equation}
	By Lemma \ref{lemma1}, we have $\| \widehat{\bxi} - \bxi^0 \|_1 = \OP(s_0 \lambda)$. In this case, $\| \bW^2_{\widehat{\bxi}} - \bW^2_{\bxi^0} \| = \OP(s_0 \lambda)$. By Assumption 5 and \citet{vershynin2010introduction}, $\lambdamax({\bX}^T \bX /n) = \OP(1)$. Thus $\| \widehat{\bSigma}_{\widehat{\bxi}} - \widehat{\bSigma}_{\bxi^0} \| = \OP( s_0 \lambda)$. 
	Therefore, after combining the two parts, we have $ \| \widehat{\bSigma}_{\widehat{\bxi}} - \bSigma_{\bxi^0} \| = \OP \{ L^2 (p/n)^{1/2} +  s_0 \lambda \}$. 
	Under ${{p}/{n}} = o(1)$ and $s_0 \lambda = o(1)$, we have $ \| \widehat{\bSigma}_{\widehat{\bxi}} - \bSigma_{\bxi^0} \| = \oP(1)$.
	
	Now for any vector $\bx$ with $\|\bx\|_2 = 1$, we have
	\[
	\displaystyle \inf_{\| {y}\|_2 = 1} \| \widehat{\bSigma}_{\widehat{\bxi}} {y}  \|_2 \le \| \widehat{\bSigma}_{\widehat{\bxi}} \bx \|_2 \le \| \bSigma_{\bxi^0} \bx \|_2 + \| (\widehat{\bSigma}_{\widehat{\bxi}} - \bSigma_{\bxi^0} ) \bx \|_2 \le \| \bSigma_{\bxi^0} \bx \|_2 + \displaystyle \sup_{\| {z}\|_2 = 1} \| (\widehat{\bSigma}_{\widehat{\bxi}} - \bSigma_{\bxi^0} ) {z} \|_2,
	\]
	which indicates that $\lambdamin(\widehat{\bSigma}_{\widehat{\bxi}}) \le \lambdamin(\bSigma_{\bxi^0}) + \| \widehat{\bSigma}_{\widehat{\bxi}} - \bSigma_{\bxi^0} \|$. Similarly, we have $\lambdamin(\bSigma_{\bxi^0})  \le  \lambdamin(\widehat{\bSigma}_{\widehat{\bxi}}) + \| \widehat{\bSigma}_{\widehat{\bxi}} - \bSigma_{\bxi^0} \|$. So $ | \lambdamin(\bSigma_{\bxi^0}) - \lambdamin(\widehat{\bSigma}_{\widehat{\bxi}}) | \le \| \widehat{\bSigma}_{\widehat{\bxi}} - \bSigma_{\bxi^0} \|$. Thus, for any $0 < \epsilon < \min \{ \| \bSigma_{\bxi^0} \|, \lambdamin( \bSigma_{\bxi^0} )/2 \}$, we have that
	\begin{eqnarray*}
		pr \left( \| \widehat{\bSigma}_{\widehat{\bxi}}^{-1} \| \ge \displaystyle \frac{1}{\lambdamin(\bSigma_{\bxi^0}) - \epsilon} \right) 
		&= & pr (  \lambdamin(\widehat{\bSigma}_{\widehat{\bxi}}) \le \lambdamin(\bSigma_{\bxi^0}) - \epsilon  )  \\
		&\le & pr ( | \lambdamin(\widehat{\bSigma}_{\widehat{\bxi}}) - \lambdamin(\bSigma_{\bxi^0})  | \ge \epsilon ) \\
		&\le & pr ( \| \widehat{\bSigma}_{\widehat{\bxi}} - \bSigma_{\bxi^0} \| \ge \epsilon ).
	\end{eqnarray*}
	Since $ \| \widehat{\bSigma}_{\widehat{\bxi}} - \bSigma_{\bxi^0} \| = \oP(1)$, we have $\| \widehat{\bSigma}_{\widehat{\bxi}}^{-1} \| = \OP(1)$. Finally, by (\ref{eq:inv}), $ \| \widehat{\bSigma}_{\widehat{\bxi}}^{-1} - \bSigma_{\bxi^0}^{-1} \| =  \OP( \| \widehat{\bSigma}_{\widehat{\bxi}} - \bSigma_{\bxi^0} \|) = \OP \{  (p/n)^{1/2} +  s_0 \lambda \}$.
\end{proof}

\begin{lemma} \label{lemma3}
	Under  Assumptions 1--3, when $p/n \rightarrow 0$, it holds that for any vector $\balpha_n \in \mathbb{R}^{p}$ with $\|\balpha_n\|_2 = 1$,
	\[
	\displaystyle \frac{{n}^{1/2} \balpha_n^T \bTheta_{\bxi^0} \Pn \dot{{\rho}}_{\bxi^0}}{({\balpha_n^T\bTheta_{\bxi^0} \balpha_n})^{1/2}}  \to N(0,1)
	\]
	in distribution as $n \to \infty$.
\end{lemma}

\begin{proof}
	We invoke the Lindeberg-Feller Central Limit Theorem. For $i = 1, \ldots, n$, let 
	\[
	Z_{ni} = \displaystyle n^{-1/2} \balpha_n^T \bTheta_{\bxi^0} \dot{{\rho}}_{\bxi^0}(y_i, \bx_i) = \displaystyle n^{-1/2} \balpha_n^T \bTheta_{\bxi^0} \bx_i \dot{\rho}(y_i, \bx_i^T \bxi^0),
	\]
	and $s_n^2 = Var \left( \sum_{i=1}^n Z_{ni} \right)$. Note that $\E  \{ \dot{\rho}(y_i, \bx_i^T \bxi^0) \mid \bx_i \} = 0$ and consequently $\E(Z_{ni})=0$. Because $\{ (y_i, \widetilde{\bx}_i)\}_{i=1}^n$ are independent and identically distributed, we can show that $s_n^2 = \balpha_n^T \bTheta_{\bxi^0} \balpha_n$.
	To show ${\sum_{i=1}^n Z_{ni}} / {s_n} \to N(0,1)$ in distribution, we first check the Lindeberg condition and then the conclusion shall follow by the Lindeberg-Feller Central Limit Theorem. Specifically, for any $\epsilon > 0$, we show that as $n \rightarrow \infty$,
	\[
	\displaystyle \frac{1}{s_n^2} \sum_{i=1}^n \E \left\{ Z_{ni}^2 \cdot {1}_{(|Z_{ni}| > \epsilon s_n)} \right\} \rightarrow 0.
	\]
	Due to the boundedness of the eigenvalues of $\bSigma_{\bxi^0}$, $\balpha_n ^T \bTheta_{\bxi^0} \balpha_n \ge \lambdamin(\bTheta_{\bxi^0}) = 1/\lambdamax(\bSigma_{\bxi^0}) \ge c_{\mathrm{max}}^{-1}$. On the other hand, by the  Cauchy-Schwarz inequality, it holds almost surely that
	\[
	\left(  \balpha_n^T \bTheta_{\bxi^0} \bx_i  \right)^2  \le  \| \balpha_n \|_2^2 \cdot \| \bTheta_{\bxi^0} \bx_i \|_2^2 
	\le  \left[   \| \bTheta_{\bxi^0} \| \cdot \|\bx_i\|_2  \right]^2 
	\le  c_{\mathrm{min}}^{-2} \cdot \mathcal{O}(pK^2).
	\]
	Inside the indicator, it holds almost surely that
	\[
	\begin{array}{rcl}
		\displaystyle \frac{Z_{ni}^2}{s_n^2} & = & \displaystyle \frac{[\dot{\rho}(y_i, \bx_i^T \bxi_0)]^2 \left(  \balpha_n^T \bTheta_{\bxi^0} \bx_i  \right)^2}{n \balpha_n ^T \bTheta_{\bxi^0} \balpha_n} \\
		& \le & [\dot{\rho}(y_i, \bx_i^T \bxi_0)]^2 \cdot c_{\mathrm{min}}^{-2} c_{\mathrm{max}} \cdot \mathcal{O}(K^2 \displaystyle \frac{p}{n}) \\
		& \le & K_1^2 c_{\mathrm{min}}^{-2} c_{\mathrm{max}} \cdot \mathcal{O}(K^2 \displaystyle \frac{p}{n}),
	\end{array}
	\]
	where the last inequality follows from the boundedness of $\dot{\rho}(y_i, \bx_i^T \bxi_0)$ in Assumption 3.  Hence, we have $Z_{ni}^2 / s_n^2 \rightarrow 0$ almost surely as $p/n \rightarrow 0$. When $n$ is large enough, $Z_{ni}^2 / s_n^2 < \epsilon^2$ and all the indicators become 0. Therefore, the Lindeberg condition holds and the Lindeber-Feller Central Limit Theorem guarantees the asymptotic normality.  
\end{proof}

\subsection{Proof of Theorem 1}

{The invertibility of $\widehat{\bSigma}_{\widehat{\bxi}}$ is shown in the proof of Lemma 2.  Now} with the bias decomposition Eq. (6) in the main text, 
\[
n^{1/2} \balpha_n^T (\widehat{\bb} - \bxi^0) - {n}^{1/2} \balpha_n^T \hTheta\bDelta = - {n}^{1/2} \balpha_n^T \hTheta \Pn \dot{{\rho}}_{\bxi^0},
\]
we first show that $\balpha_n^T \hTheta \balpha_n - \balpha_n^T {\bTheta}_{\bxi^0} \balpha_n = o_{P}(1)$ and that $$ {{n}^{1/2}\balpha_n^T \hTheta \Pn \dot{{\rho}}_{\bxi^0}} / {({\balpha_n^T \hTheta \balpha_n})^{1/2}} = {{n}^{1/2}\balpha_n^T {\bTheta}_{\bxi^0} \Pn \dot{{\rho}}_{\bxi^0}} / {({\balpha_n^T {\bTheta}_{\bxi^0} \balpha_n})^{1/2}} + o_{P}(1),$$ Then by Slutsky's Theorem, the asymptotic distribution of the target ${{n}^{1/2}\balpha_n^T \hTheta \Pn \dot{{\rho}}_{\bxi^0}} / {({\balpha_n^T \hTheta \balpha_n})^{1/2}}$ can be derived by using the asymptotic distribution of  ${{n}^{1/2} \balpha_n^T {\bTheta}_{\bxi^0} \Pn \dot{{\rho}}_{\bxi^0}} / {({\balpha_n^T {\bTheta}_{\bxi^0} \balpha_n})^{1/2}}$, which has been proved in Lemma \ref{lemma3}. In the final step, as long as ${n}^{1/2} \balpha_n^T \hTheta \bDelta = \oP(1)$, the asymptotic distribution of ${{n}^{1/2} \balpha_n^T(\widehat{\bb} - \bxi^0)} / {({\balpha_n^T \hTheta \balpha_n})^{1/2}}$ follows immediately.

According to Lemma \ref{lemma2}, it follows that 
\[
\vert \balpha_n^T \hTheta \balpha_n - \balpha_n^T {\bTheta}_{\bxi^0} \balpha_n \vert = \vert \balpha_n^T  (\hTheta - {\bTheta}_{\bxi^0} ) \balpha_n \vert \le \| \hTheta - {\bTheta}_{\bxi^0} \| \cdot \| \balpha_n \|_2^2 =   \oP(1).
\]
By the Cauchy-Schwartz inequality,
\[
\begin{array}{rcl}
	{n}^{1/2} \vert \balpha_n^T \hTheta \Pn \dot{{\rho}}_{\bxi^0} - \balpha_n^T {\bTheta}_{\bxi^0} \Pn \dot{{\rho}}_{\bxi^0} \vert & \le & {n}^{1/2} \| \balpha_n \|_2 \cdot \| (\hTheta - \bTheta_{\bxi^0}) \Pn \dot{{\rho}}_{\bxi^0} \|_2 \\
	& \le & {n}^{1/2} \| \hTheta - \bTheta_{\bxi^0} \| \cdot \| \Pn \dot{{\rho}}_{\bxi^0} \|_2,
\end{array}
\]
then we have 
\[ \begin{array}{rcl}
	{n}^{1/2} | \balpha_n^T \hTheta \Pn \dot{{\rho}}_{\bxi^0} - \balpha_n^T {\bTheta}_{\bxi^0} \Pn \dot{{\rho}}_{\bxi^0} | & \le & {n}^{1/2}  \cdot \|  \Pn \dot{{\rho}}_{\bxi^0} \|_{2} \cdot \OP \{ \displaystyle  ({p}/{n})^{1/2} + s_0 \lambda \}. \\
\end{array}
\]
By definition, 
\[
\begin{array}{rcl}
	\| \Pn \dot{{\rho}}_{\bxi^0} \|_{2}^2  & = & \displaystyle \sum_{j=1}^p \{ n^{-1} \sum_{i=1}^n x_{ij} \dot{\rho} (y_i, x_i^T \xi^0) \}^2 \\
	& = & \displaystyle n^{-2} \sum_{j=1}^p \sum_{i=1}^n \sum_{k=1}^n x_{ij} x_{kj} \dot{\rho}(y_i, x_i^T \xi^0) \dot{\rho}(y_k, x_k^T \xi^0).
\end{array}
\]
With independent observations  
and $\E\{ x_{ij} \dot{\rho}(y_i, x_i^T \xi^0) \} = 0$ for any $i$, it follows that
\[
\E \| \Pn \dot{\rho}_{\xi^0} \|^2_2 = \frac{1}{n^2} \sum_{j=1}^p \sum_{i=1}^n \E \{ x_{ij}^2 \dot{\rho}^2 (y_i, x_i^T \xi^0) \}. 
\]
By Assumptions 1 and  3, we have $| x_{ij} \dot{\rho}(y_i, \bx_i^T \bxi^0) | \le K K_1$ almost surely holds for all $i$ and $j$, so $\E \| \Pn \dot{\rho}_{\xi^0} \|^2_2 = \mathcal{O}(p/n)$. This implies that $\| \Pn \dot{\rho}_{\xi^0} \|_2 = \OP\{ ({p/n})^{1/2} \}$.
Then we have 
\[
{n}^{1/2} | \balpha_n^T \hTheta \Pn \dot{{\rho}}_{\bxi^0} - \balpha_n^T {\bTheta}_{\bxi^0} \Pn \dot{{\rho}}_{\bxi^0} | \le 
\OP \left(  p/{n}^{1/2} +  s_0 \lambda {p}^{1/2} \right),
\]
which is $\oP(1)$ by our assumption in Theorem 1.

Finally, we prove $| {n}^{1/2} \balpha_n^T \hTheta \bDelta| = \oP(1)$. By the Cauchy-Schwartz inequality, $| {n}^{1/2} \balpha_n^T \hTheta \bDelta| \le {n}^{1/2} \| \hTheta \bDelta \|_2$, we only need that ${n}^{1/2} \| \hTheta \bDelta \|_2 = \oP(1)$. Note that
\[
\Delta_j  = \displaystyle \frac{1}{n} \sum_{i=1}^n  \left\{  \ddot{\rho}(y_i, a_i^*) - \ddot{\rho}(y_i, \bx_i^T \widehat{\bxi}~)  \right\} x_{ij} \bx_i^T (\bxi^0 - \widehat{\bxi}~),
\]
where $a_i^*$ lies between $\bx_i^T\widehat{\bxi}$ and $\bx_i^T\bxi^0$, i.e. $ | a_i^* - \bx_i^T \widehat{\bxi}~ | \le | \bx_i^T ( \widehat{\bxi} - \bxi^0 ) | $. Then uniformly for all $j$,
\[
\begin{array}{rcl}
	| \Delta_j | & \le & \displaystyle \frac{1}{n} \sum_{i=1}^n  \vert \ddot{\rho}(y_i, a_i^*) - \ddot{\rho}(y_i, \bx_i^T \widehat{\bxi}) \vert \cdot |x_{ij}| \cdot  | \bx_i^T(\bxi^0 - \widehat{\bxi})| \\
	& \le & \displaystyle \frac{1}{n} \sum_{i=1}^n c_{Lip} | a_i^* - \bx_i^T \widehat{\bxi}   |  \cdot K \cdot | \bx_i^T(\bxi^0 - \widehat{\bxi})| \\
	& \le & \displaystyle c_{Lip}   K \cdot \frac{1}{n} \sum_{i=1}^n  | \bx_i^T(\bxi^0 - \widehat{\bxi})|^2 \\ 
	& = &  c_{Lip}   K \cdot  \OP(  s_0 \lambda^2) \\
	& = &  \OP(  s_0 \lambda^2),
\end{array}
\]
where the last equality holds by Lemma \ref{lemma1}. Since $\|  \bTheta_{\bxi^0}\| = \mathcal{O}(1)$ and $\| \hTheta - \bTheta_{\bxi^0} \| = \oP(1)$, it follows that $\| \hTheta \| = \OP(1)$, and 
\[
\begin{array}{rcl}
	{n}^{1/2} \| \hTheta \bDelta \|_2 &  \le & {n}^{1/2} \| \hTheta \| \cdot \| \bDelta \|_{2} \\
	& \le & {n}^{1/2} \OP(1) \cdot  {p}^{1/2} \|\bDelta\|_{\infty} \\
	& \le & \OP( ({np})^{1/2}  s_0 \lambda^2 ).
\end{array}
\]
By the assumption of $(np)^{1/2} s_0\lambda^2 = o(1)$ in Theorem 1, ${n}^{1/2} \|\hTheta\bDelta\|_2 = \oP(1)$. 
Applying Slutsky's Theorem and Lemma \ref{lemma3} gives the result.

{
	Part (ii) in Theorem 1 can be proved using Cram\'{e}r-Wold device. For any $\widetilde{\ba} \in \mR^{m}$, let $\balpha_n = \bA^T_n \widetilde{\ba}$ in Theorem 1(i), which would still hold when $\| \balpha_n \|_2 \le c'$ for some constant $c' >0$. In this case, $\| \balpha_n \|_2 = \| \bA_n^T \widetilde{\ba} \|_2 \le \| \bA_n^T \| \| \widetilde{\ba} \|_2 \le c_* \| \widetilde{\ba} \|_2$
	is upper bounded by a constant since $\widetilde{\ba}$ has a fixed dimension. Then, as $n \to \infty$,
	\[
	\displaystyle \frac{n^{1/2}\widetilde{\ba}^T \bA_n (\widehat{\bb} - \bxi^0)}{(\widetilde{\ba}^T \bA_n \widehat{\bTheta} \bA_n^T \widetilde{\ba} )^{1/2}} \overset{\mathcal{D}}{\to} N(0,1).
	\]
	The variance
	$
	| \widetilde{\ba}^T \bA_n \widehat{\bTheta} \bA_n^T \widetilde{\ba} - \widetilde{\ba}^T \bA_n \Thetabeta \bA_n^T \widetilde{\ba} | \le \| \widehat{\bTheta} - \Thetabeta \| \| \bA_n^T \widetilde{\ba} \|_2^2 = \oP(1). 
	$ Hence, by Slutsky's Theorem, 
	\[
	n^{1/2}\widetilde{\ba}^T \bA_n (\widehat{\bb} - \bxi^0) \overset{\mathcal{D}}{\to} N(0, \widetilde{\ba}^T \bF \widetilde{\ba} ).
	\]
}

\section{Additional Simulations}

\subsection{Simulation studies: large $n$, diverging $p$}

We examined the scenario with smaller sample sizes, where we simulated $n=500$ observations with $p = 20, 100, 200, 300$ covariates in logistic regression models. The rest of the settings were identical to those with $n=1000$ in the main text. Figures \ref{fig:logit_n500_rho7_indep} -- \ref{fig:logit_n500_rho7_cs} display the results from three types of covariance structures, including the identity matrix, the autoregressive structure of order 1 or AR(1) with correlation 0.7, and the compound symmtry structure with correlation 0.7, respectively.  Figure \ref{fig:logit_n500_rho7_cs} shows that with $n=500$, $p=300$ and the compound symmetry structure, neither of the de-biased lasso methods worked well, which is not surprising given the relatively small sample size and highly correlated covariates.

We also varied the correlation $\rho=0.2$ in the covariance matrix $\bSigma_x$ for the autoregressive and compound symmetry structures to reflect the presence of less correlated covariates; see Figure~\ref{fig:logit_n500_rho2_ar1} and Figure~\ref{fig:logit_n500_rho2_cs}, respectively. These results are close to the independent covariate case. To summarize, our proposed refined de-biased lasso approach, in most cases, can provide the best bias correction and honest confidence intervals.

In Section 4, additional simulation results have been shown to demonstrate that generally $\mu_n=0$ leads to the best performance empirically in Eq.~(5). In the presented logistic regression setting, we simulate $n=500$ observations and $p=40, 100, 200, 300, 400$ covariates for 200 times. Covariates follow a multivariate Gaussian distribution with mean zero and AR(1) covariance matrix ($\rho=0.7$). Only two coefficients are non-zero (1 and 0.5) and the rest are noises. We pre-specify a sequence of values in $[0,1]$ for the tuning parameter $\mu_n$ in Eq. (5), equally spaced in log scale. The de-biased lasso estimator based on Eq. (5) is referred to by ``tuning". 

Figure \ref{fig:tuning_large2} (already shown in Section 4 of the main article) and Figure \ref{fig:tuning_small} show the simulation results for the coefficients $\xi_j^0=1$ and $\xi_j^0=0.5$ respectively, where the three columns correspond to average estimation bias, coverage probability for its 95\% confidence interval, and the ratio between its average model-based standard error and empirical standard error, over 200 replications. Since our main focus is good bias correction and honest confidence interval coverage, we find that over a very wide range of number of covariates, $\mu_n=0$ performs the best empirically.

\begin{figure}
	\centering
	\includegraphics[width=\textwidth]{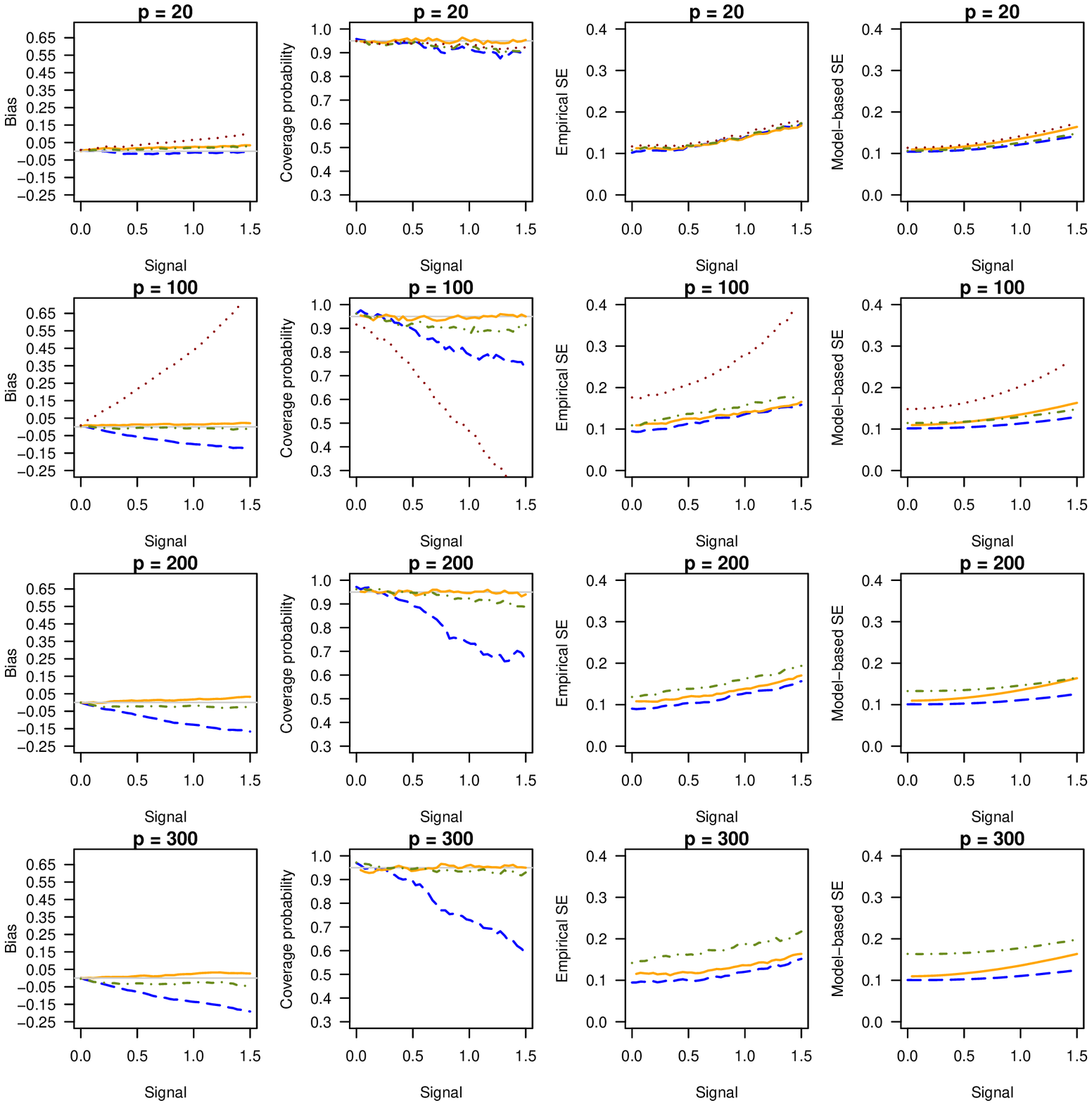}
	\caption{Simulation results: Bias, coverage probability, empirical standard error, and model-based standard error for $\beta_1^0$ in a logistic regression. Covariates are simulated from $N_{p} (0_{p}, \bI)$ before being truncated at $\pm 6$. The sample size is $n=500$ and the number of covariates $p = 20, 100, 200, 300$. The oracle estimator, that is the maximum likelihood estimator under the true model, is plotted as a reference in orange solid lines. The methods in comparisons include our proposed refined de-biased lasso in olive dot-dash lines, the original de-biased lasso by \citet{van2014asymptotically} in blue dashed lines, and the maximum likelihood estimation in red dotted lines.}
	\label{fig:logit_n500_rho7_indep}
\end{figure}

\begin{figure}
	\centering
	\includegraphics[width=\textwidth]{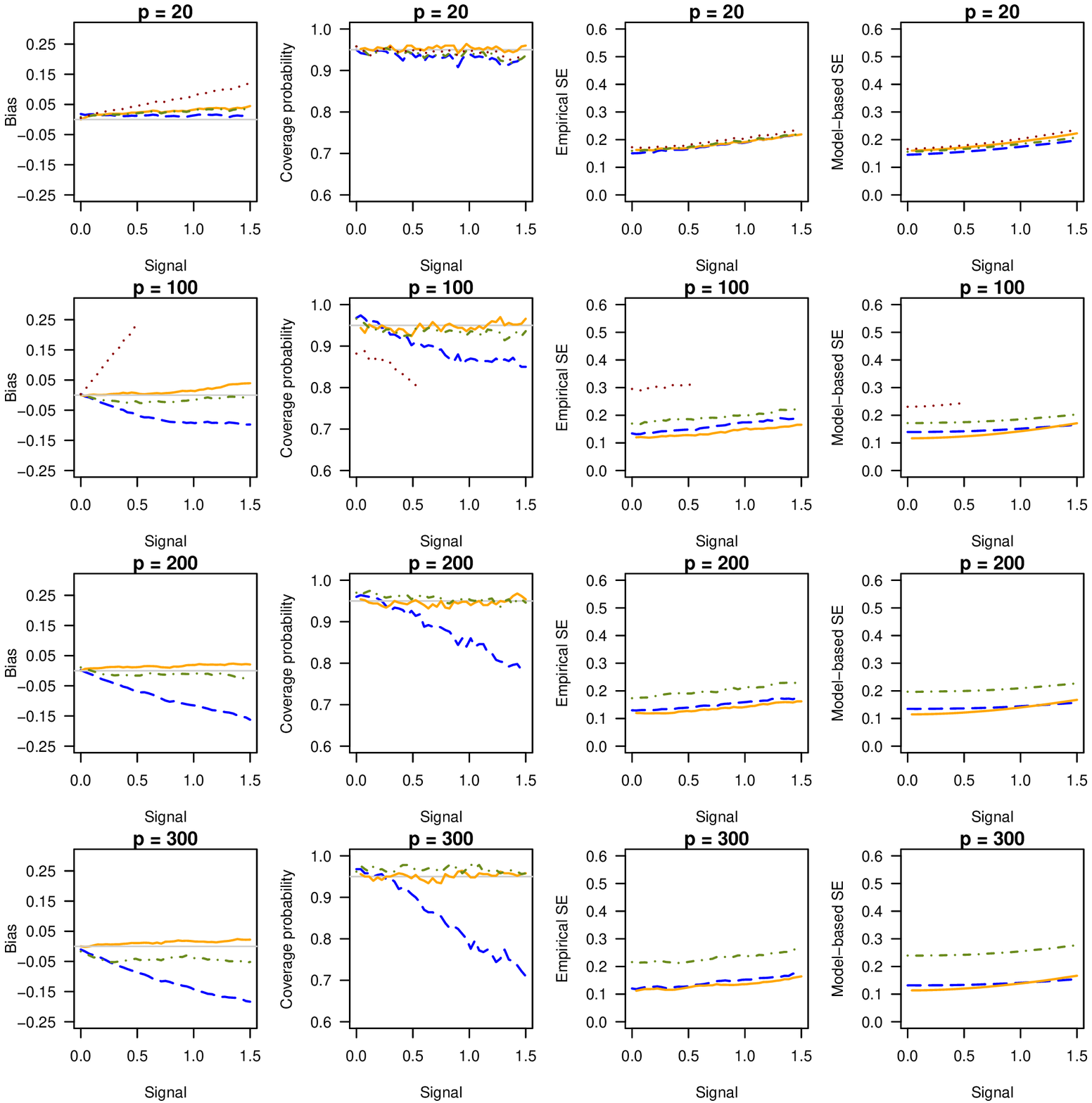}
\caption{Simulation results: Bias, coverage probability, empirical standard error, and model-based standard error for $\beta_1^0$ in a logistic regression. Covariates are simulated from $N_{p} (0_{p}, \bSigma_x)$ before being truncated at $\pm 6$, where $\Sigma_x$ has an autoregressive covariance structure of order 1 with $\rho=0.7$. The sample size is $n=500$ and the number of covariates $p = 20, 100, 200, 300$. The oracle estimator, that is the maximum likelihood estimator under the true model, is plotted as a reference in orange solid lines. The methods in comparisons include our proposed refined de-biased lasso in olive dot-dash lines, the original de-biased lasso by \citet{van2014asymptotically} in blue dashed lines, and the maximum likelihood estimation in red dotted lines.}
\label{fig:logit_n500_rho7_ar1}
\end{figure}

\begin{figure}
\centering
\includegraphics[width=\textwidth]{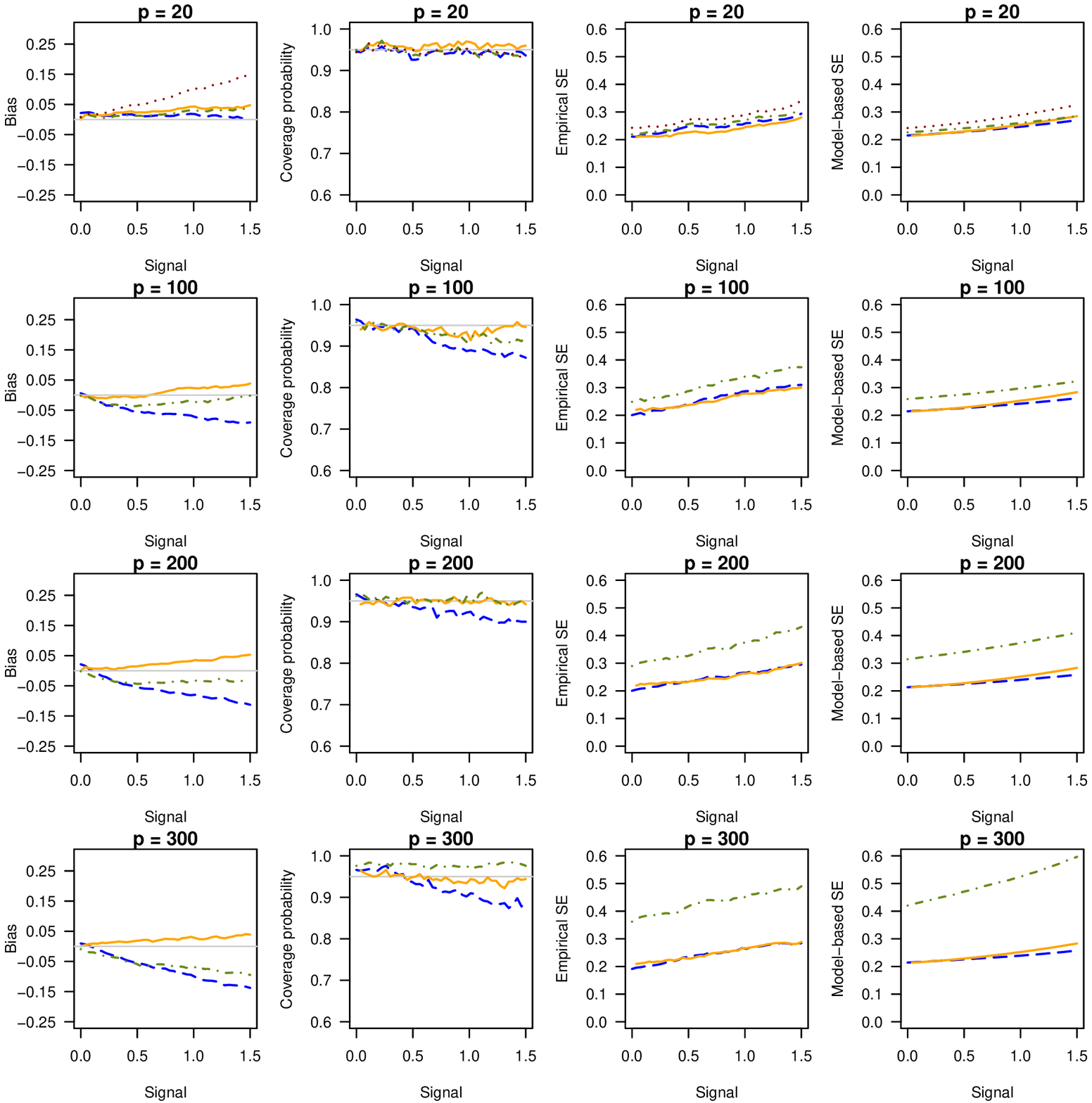}
\caption{Simulation results: Bias, coverage probability, empirical standard error, and model-based standard error for $\beta_1^0$ in a logistic regression.  Covariates are simulated from $N_{p} (0_{p}, \bSigma_x)$ before being truncated at $\pm 6$, where $\Sigma_x$ has a compound symmetry structure with $\rho=0.7$.  The sample size is $n=500$ and the number of covariates $p = 20, 100, 200, 300$. The oracle estimator, that is the maximum likelihood estimator under the true model, is plotted as a reference in orange solid lines. The methods in comparisons include our proposed refined de-biased lasso in olive dot-dash lines, the original de-biased lasso by \citet{van2014asymptotically} in blue dashed lines, and the maximum likelihood estimation in red dotted lines.}
\label{fig:logit_n500_rho7_cs}
\end{figure}

\begin{figure}
\centering
\includegraphics[width=\textwidth]{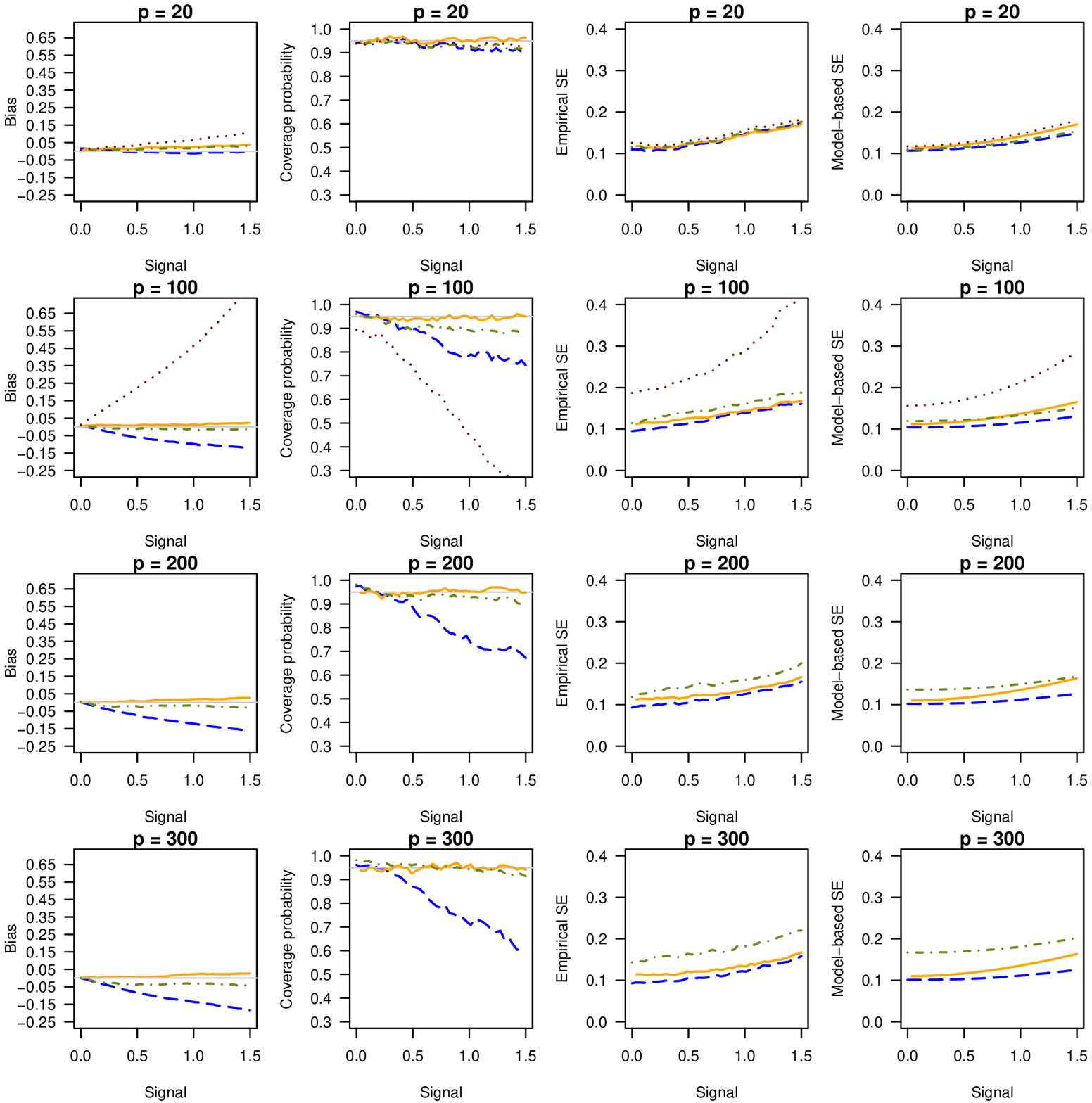}
\caption{Simulation results: Bias, coverage probability, empirical standard error, and model-based standard error for $\beta_1^0$ in a logistic regression. Covariates are simulated from $N_{p} (0_{p}, \bSigma_x)$ before being truncated at $\pm 6$, where $\Sigma_x$ has an autoregressive covariance structure of order 1 with $\rho=0.2$. The sample size is $n=500$ and the number of covariates $p = 20, 100, 200, 300$. The oracle estimator, that is the maximum likelihood estimator under the true model, is plotted as a reference in orange solid lines. The methods in comparisons include our proposed refined de-biased lasso in olive dot-dash lines, the original de-biased lasso by \citet{van2014asymptotically} in blue dashed lines, and the maximum likelihood estimation in red dotted lines.}
\label{fig:logit_n500_rho2_ar1}
\end{figure}

\begin{figure}
\centering
\includegraphics[width=\textwidth]{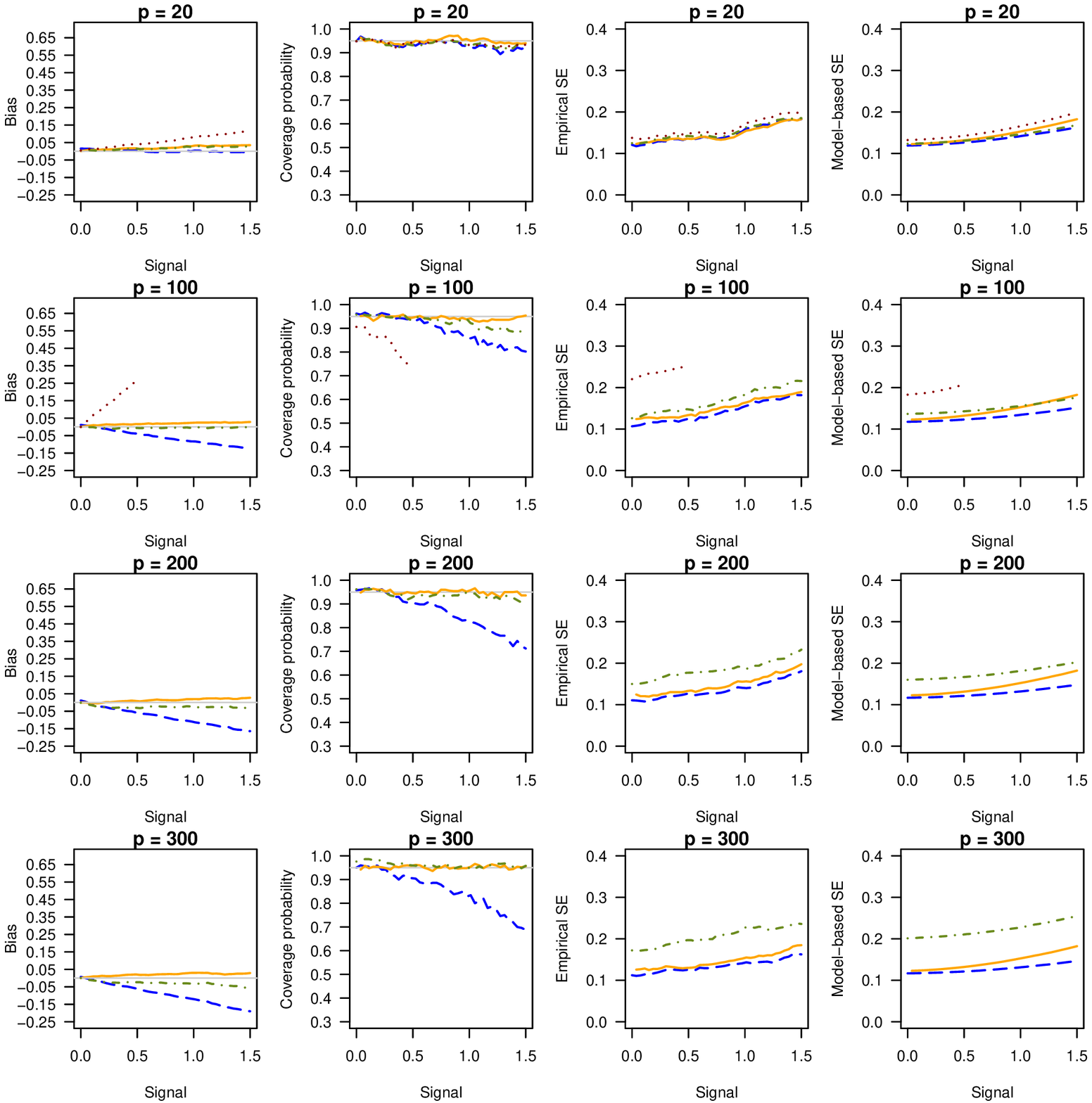}
\caption{Simulation results: Bias, coverage probability, empirical standard error, and model-based standard error for $\beta_1^0$ in a logistic regression.  Covariates are simulated from $N_{p} (0_{p}, \bSigma_x)$ before being truncated at $\pm 6$, where $\Sigma_x$ has a compound symmetry structure with $\rho=0.2$.  The sample size is $n=500$ and the number of covariates $p = 20, 100, 200, 300$. The oracle estimator, that is the maximum likelihood estimator under the true model, is plotted as a reference in orange solid lines. The methods in comparisons include our proposed refined de-biased lasso in olive dot-dash lines, the original de-biased lasso by \citet{van2014asymptotically} in blue dashed lines, and the maximum likelihood estimation in red dotted lines.}
\label{fig:logit_n500_rho2_cs}
\end{figure}

\begin{figure}
	\centering	
	\includegraphics[width=0.9\textwidth]{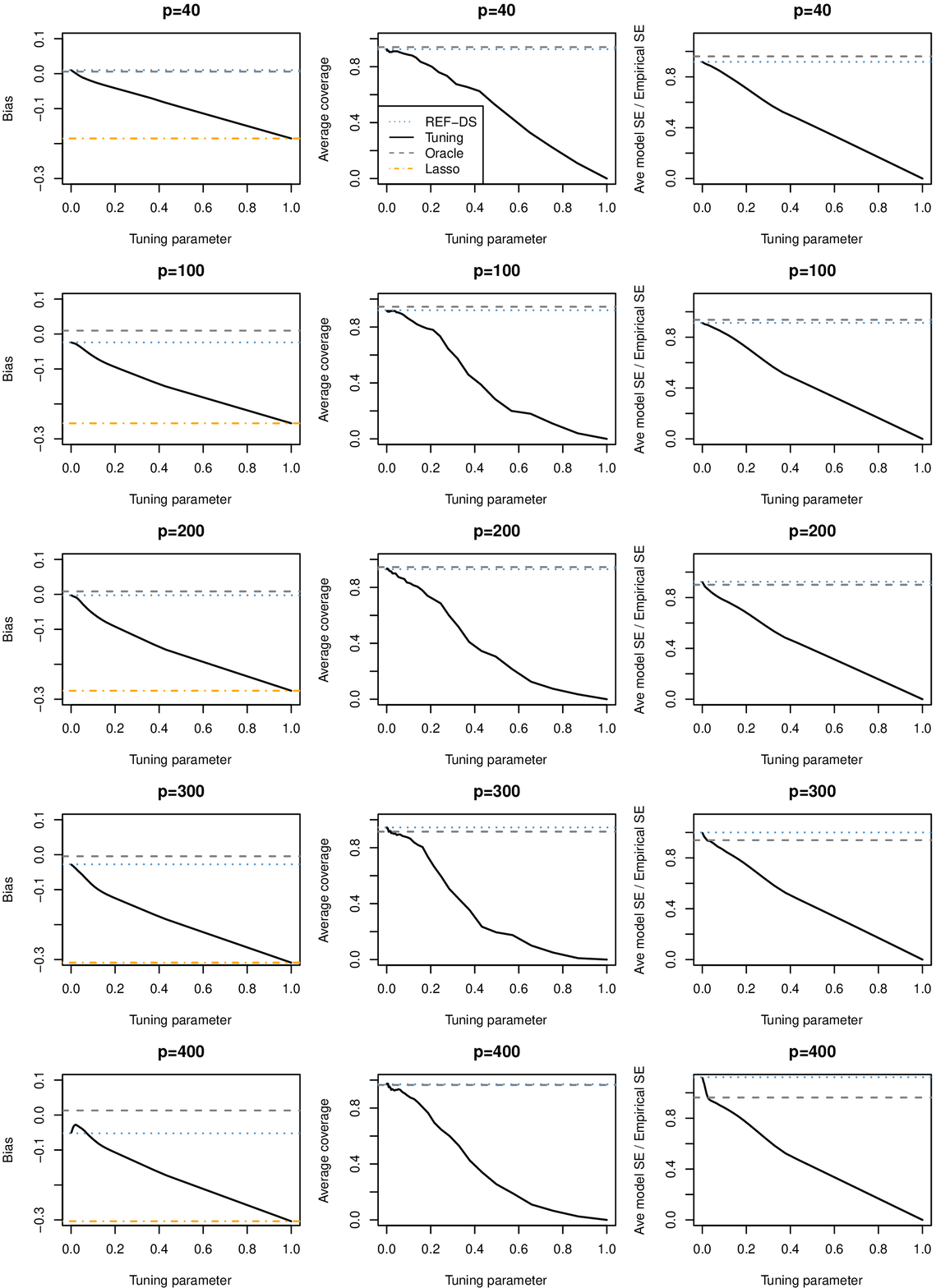}
	\caption{Simulation results that verify the selection of the tuning parameter $\mu_n=0$ in Eq. (5) for $\xi_j^0=1$.}
	\label{fig:tuning_large2}
\end{figure}

\begin{figure}
\centering
\includegraphics[width=0.9\textwidth]{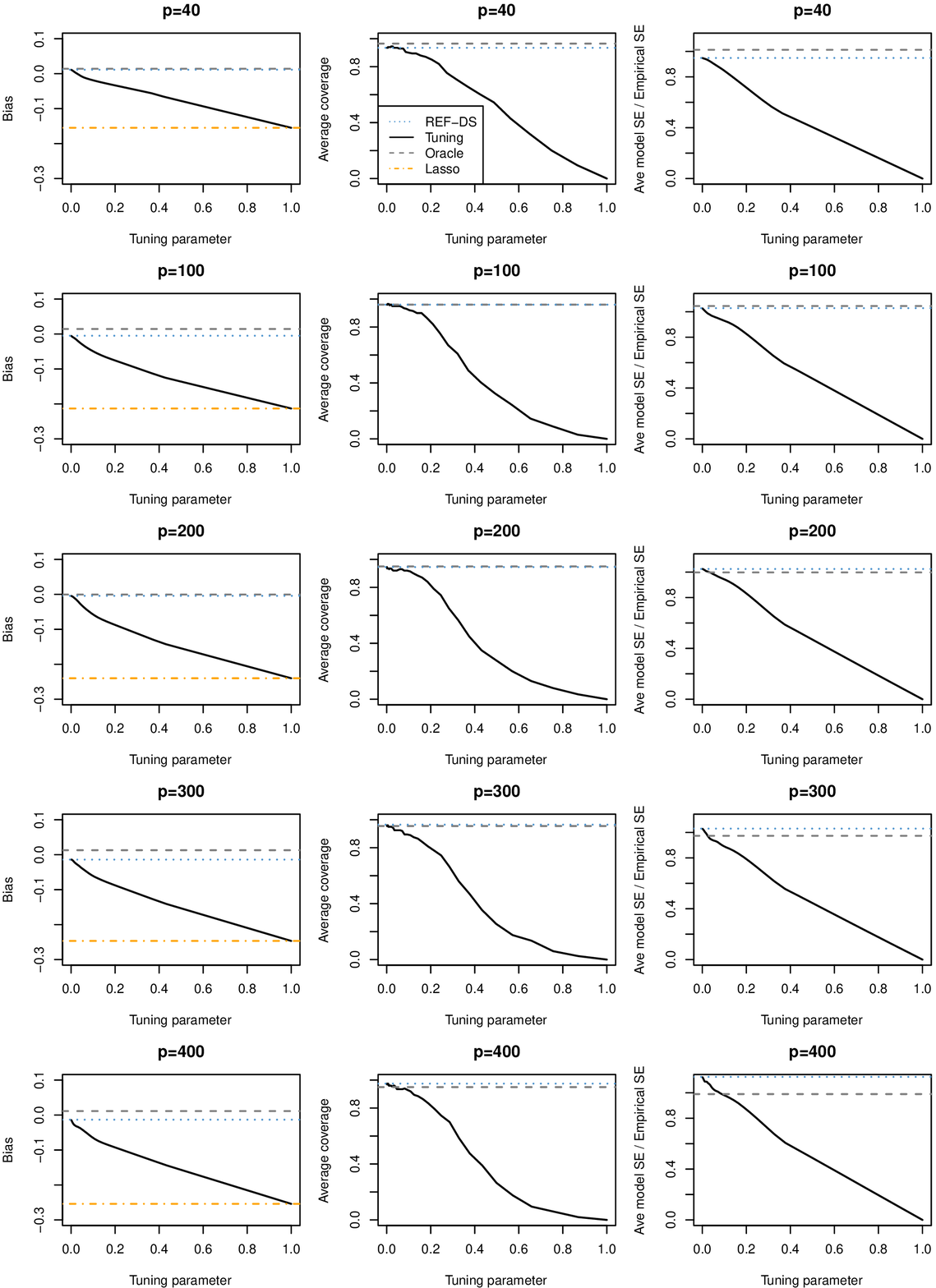}
\caption{Simulation results that verify the selection of the tuning parameter $\mu_n=0$ in Eq. (5) for $\xi_j^0=0.5$.}
\label{fig:tuning_small}
\end{figure}

\subsection{Simulation studies: large $p$, small $n$}

We also present simulation studies that feature logistic regression models in the ``large $p$, small $n$" setting, with $n=300$ observations and $p=500$ covariates. For simplicity, covariates are simulated from $N_{p}({0}, \bSigma_x)$, where $\bSigma_{x, ij} = 0.7^{|i-j|}$, and truncated at $\pm 6$. In the true coefficient vector $\bbeta^0$, the intercept $\beta^0_0 = 0$ and $\beta^0_1$ varies from 0 to 1.5 with 40 equally spaced increments. To examine the impacts of different true model sizes, we arbitrarily choose $\bar{s}_0 = $2, 4 or 10 additional coefficients from the rest in $\bbeta^0$, and fix them at 1 throughout the simulation. At each value of $\beta_1^0$, a total of 500 simulated datasets are generated. We focus on the de-biased estimates and inference for $\beta_1^0$ using the method of  \citet{van2014asymptotically}. 

Figure \ref{fig:sim_largep}, with the true model size increasing from the top to the bottom, shows that the de-biased lasso estimate for $\beta_1^0$ has a bias which almost linearly increases with the true  size of $\beta_1^0$. This undermines the credibility of the consequent confidence intervals. Meanwhile, the model-based variance overestimates the true variance for smaller signals and underestimates it for larger signals in the two  models with smaller model sizes, as shown by the top two rows in Figure \ref{fig:sim_largep}. This partially explains the over- and under-coverage for smaller and larger signals, respectively. Due to penalized estimation in node-wise lasso, the  variance of the original de-biased lasso estimator is even smaller than the oracle maximum likelihood estimator obtained as if the true model were known; see the bottom two rows in Fig. \ref{fig:sim_largep}. The empirical coverage probability decreases to about 50\% as the signal $\beta_1^0$ goes to 1.5, and when the true model size reaches 5;  see the middle row in Figure \ref{fig:sim_largep}. The bias correction is sensitive to the true model size, which becomes worse for larger true models. We have also conducted simulations by changing the covariance structure of covariates to be  independent or compound symmetry with correlation coefficient 0.7 and variance 1, and  have obtained similar results.

\begin{figure}
	\centering
	\includegraphics[width=0.7\textwidth]{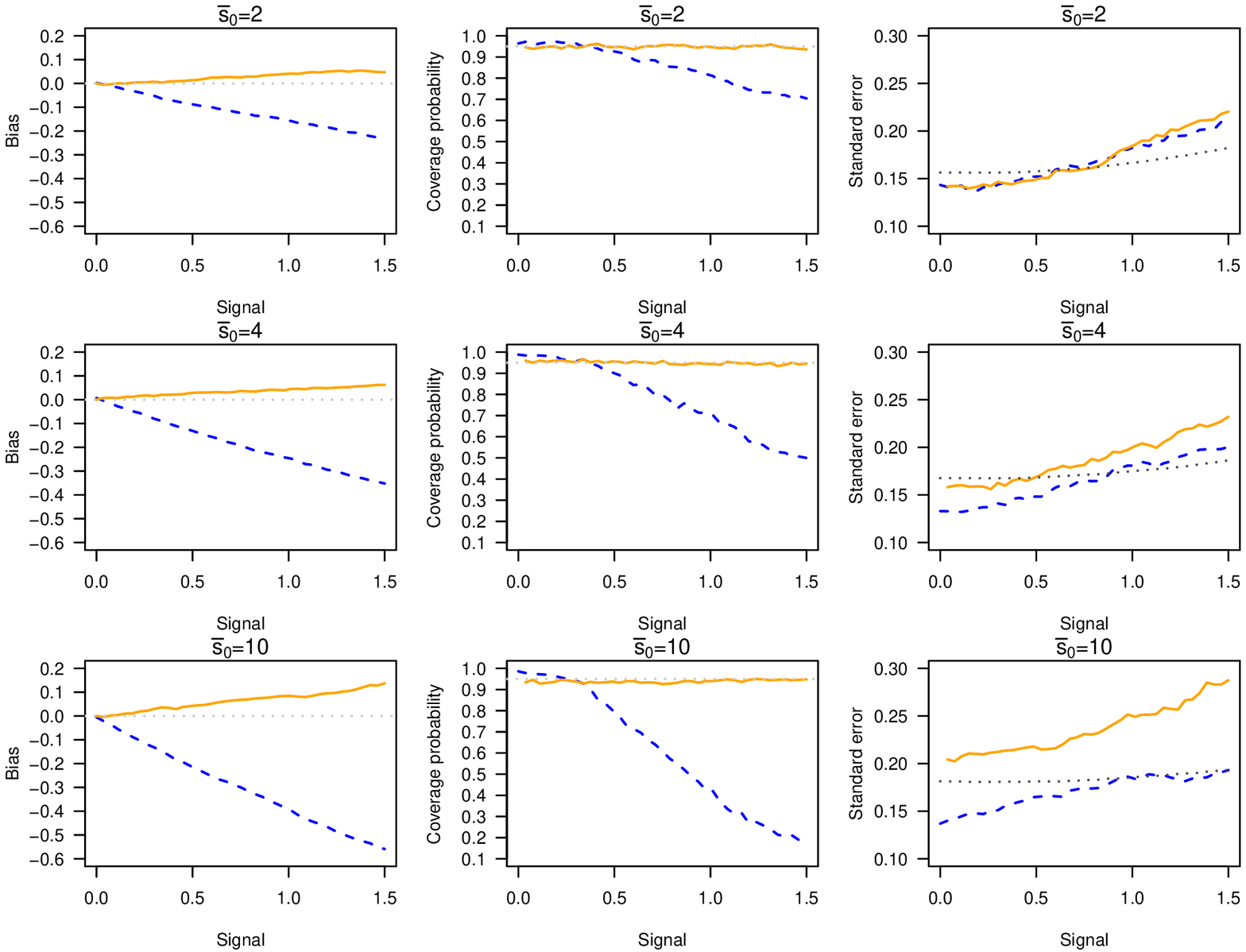}
	\caption{Simulation results of a logistic regression with sample size $n=300$ and $p=500$ covariates. Covariates are simulated from $N_{p} (0_{p}, \bSigma_x)$ before being truncated at $\pm 6$, where $\Sigma_x$ has an autoregressive covariance structure of order 1 with $\rho=0.7$.  The left column  presents estimation bias, the middle column presents empirical coverage probability, and the right column presents standard error, both model-based and empirical, of the estimated $\beta_1^0$. Horizontal panels correspond to models with 2, 4 and 10 additional signals fixed at 1 from the top to the bottom, respectively. In the left and middle columns, blue dashed lines represent the original de-biased lasso approach by \citet{van2014asymptotically}, and orange solid  lines represent the oracle estimator. In the right column, blue dashed lines and black dotted lines represent the empirical standard error and the model-based standard error from the method of \citet{van2014asymptotically}, respectively, and orange solid lines for the empirical standard error of the oracle estimator.}
	\label{fig:sim_largep}
\end{figure}

\section{Demographics of the Boston Lung Cancer Survivor Cohort}

Table \ref{tab:blcs_pop} summarizes the demographics of the 1,374 individuals studied in the main text, stratified by their smoking status. 

\begin{table}
	\centering
	\caption{Characteristics of the individuals in the analytical data set of the Boston Lung Cancer Survivor Cohort} 	\label{tab:blcs_pop}
	\small
	\begin{tabular}{lccc}
		\hline
		\multirow{2}{*}{Information} & {Overall} & {Among smokers} & {Among non-smokers} \\
		&  Count (\%) / Mean (SD\footnote{Standard deviation}) & Count (\%) / Mean (SD) & Count (\%) / Mean (SD) \\  \hline
		Total & 1374 (100\%) & 1077 (100\%) & 297 (100\%) \\
		Lung cancer &  &  &  \\
		$\quad$ Yes & 651 (47.4\%) & 595 (55.2\%) & 56 (18.9\%) \\
		$\quad$ No & 723  (52.6\%) & 482  (44.8\%) & 241   (81.1\%) \\
		Education & & & \\
		$\quad$ No high school & 153    (11.1\%) & 139   (12.9\%)  &  14    (4.7\%) \\
		$\quad$ High school graduate & 374 (27.2\%) & 309 (28.7\%) &  65 (21.9\%) \\
		$\quad$ At least 1-2 years of college & 847 (61.7\%) & 629 (58.4\%) & 218 (73.4\%) \\
		Gender & & & \\
		$\quad$ Female & 845  (61.5\%) & 644  (59.8\%) &  201   (67.7\%) \\
		$\quad$ Male & 529 (38.5\%) & 433 (40.2\%)  & 96 (32.3\%) \\
		Age & 60.0 (10.6) & 60.7 (10.2) & 57.7 (11.7) \\
		\hline
	\end{tabular}
\end{table}

\section{Discussion on the difference between sparsity assumptions in our work and \citet{van2014asymptotically}}

We first notice there are two kinds of sparsity parameters: one for the sparsity of
regression coefficients (denoted by $s_0$), and the other for 
the sparsity of the inverse of the information matrix, $\Thetabeta$ (denoted by
$s_j$, the number of non-zero elements in the $j$th row of $\Thetabeta$).
The following clarifies the extent to which the assumptions of  our Theorem 1  differs from those of \citet{van2014asymptotically}.
For the model sparsity $s_0$, our assumption $s_0 \log(p) (p/n)^{1/2} \to 0$ is indeed more stringent than  $s_0 \log(p) / \sqrt{n} \to 0$  required by \citet{van2014asymptotically}, whereas for the sparsity of the inverse information matrix, \citet{van2014asymptotically} assumed $s_j = o(\sqrt{n/\log(p)})$ for all $j$ and we do not make any assumptions on $s_j$ directly.  
A related condition set by us is $p^2/n \to 0$, which is weaker than $s_j = o(\sqrt{n/\log(p)})$ by a logarithmic factor if $s_j \asymp p$. 

Below we elaborate on how these  sparsity differences lead to different results  obtained by  our manuscript and  \citet{van2014asymptotically}, which indeed have different inferential objectives, and the rationale why our  assumptions fit our
inferential objectives. 

First, these two works differ in  inferential objectives. We aim to infer  any linear combinations of the regression parameter, i.e. $\alpha_n^T \xi^0$, where the only constraint on $\alpha_n$ is $\| \alpha_n \|_2 = 1$ (in fact, bounded $\| \alpha_n \|_2$ would suffice). Thus, we have to control the behavior of the  $(p+1)\times (p+1)$ matrix $\{ \widehat{\Theta} - \Thetabeta \}$. In contrast, \citet{van2014asymptotically}  inferred individual   components in $\xi^0$ one at a time, making it sufficient to control the rates of $\{ \widehat{\Theta}_j - \Theta_{\xi^0,j} \}$ (here the subscript $j$ indicates the $j$th row of a matrix) for one row at a time, and the node-wise lasso  provides such required rates. 


Second, besides the essential assumptions that both papers require (our Assumptions 1--4), \citet{van2014asymptotically} has another important assumption that we do not need to assume, that is, $\| \mathbf{X}_{\beta^0,-j} \gamma^0_{\beta^0,j} \|_{\infty} = \mathcal{O}(1)$ (see their Theorem 3.3 (iv)), which results in $\| \mathbf{X}\widehat{\Theta}_j^T \|_{\infty} = \OP(K)$ in their condition (C5). In our notation, this assumption would be equivalent to the boundedness on $\| \Thetabeta x_i \|_{\infty}$ and would result in $\| \widehat{\Theta} x_i \|_{\infty}$ being bounded in probability. However, we have elected {not to directly make such assumptions on the inverse of the informative matrix and its estimate} as they {may be closely related to the sparsity requirement of $\Thetabeta$ under Assumption 1, and may not hold or be verifiable in GLM settings.} 

Finally, we clarify that our specified order assumption on $s_0$ with respect to $n$ and $p$ is to ensure  $| n^{1/2} \alpha_n^T \widehat{\Theta} \Delta | = \oP(1)$ in the proof of Theorem 1 (please see Pages 5--6 in Web Appendix A), which is 
for inference on any linear combinations of regression coefficients. However, if we had aimed for 
a weaker result of inferring an individual coefficient only as in \citet{van2014asymptotically},
we would have let $\alpha_n = e_j$ (a $p$-dimensional vector with the $j$th element being 1 and  all the other elements being
zero) corresponding to drawing inference on the effect of the $j$th covariate, and also with a condition of $\| \widehat{\Theta} x_i \|_{\infty} = \OP(1)$ as in \citet{van2014asymptotically}, 
we would have had
\[
\begin{array}{rl}
	| \sqrt{n} \widehat{\Theta}_j \Delta | & = | \sqrt{n} \displaystyle \frac{1}{n} \sum_{i=1}^n \left\{ \ddot{\rho}(y_i, a_i^*) - \ddot{\rho}(y_i, x_i^T \widehat{\xi}) \right\} \widehat{\Theta}_j x_i x_i^T (\xi^0 - \widehat{\xi})  | \\
	& \le \sqrt{n} \displaystyle \frac{1}{n} \sum_{i=1}^n |\ddot{\rho}(y_i, a_i^*) - \ddot{\rho}(y_i, x_i^T \widehat{\xi})| \cdot |\widehat{\Theta}_j x_i| \cdot |x_i^T (\xi^0 - \widehat{\xi})| \\
	& \le \sqrt{n} \displaystyle \frac{1}{n} \sum_{i=1}^n c_{Lip} |x_i^T (\xi^0 - \widehat{\xi})| \cdot \OP(1) \cdot |x_i^T (\xi^0 - \widehat{\xi})| \\
	& = \sqrt{n} \OP(1) \displaystyle \frac{1}{n} \sum_{i=1}^n |x_i^T (\xi^0 - \widehat{\xi})|^2 \\
	& = \OP(\sqrt{n} s_0 \lambda^2).
\end{array}
\]
Therefore, to infer $\xi_j^0$ alone, we would have reached the  same assumption that  $s_0 \log(p) / \sqrt{n} \to 0$ as in \citet{van2014asymptotically} with $\lambda \asymp \sqrt{\log(p)/n}$.

In summary,  our work may be meritorious by  providing readers with these explicit rates for guaranteeing proper inferences when directly inverting the information matrix.


\vspace*{-8pt}

\end{document}